\documentclass[10pt]{amsart} 
\usepackage{amsmath, amsthm, amssymb}
\usepackage{graphicx}
\usepackage{mathbbol}
\usepackage{txfonts}
\usepackage{fullpage}
\usepackage[mathcal]{euscript} 
\usepackage{dsfont, pxfonts, verbatim} 
\numberwithin{equation}{section}
        \newtheorem{theorem}{Theorem}[section]
        \newtheorem{proposition}[theorem]{Proposition}
        \newtheorem{lemma}[theorem]{Lemma}
         
        \newtheorem{definition}[theorem]{Definition} 
        \newtheorem{remark}[theorem]{Remark}  
        
 \newtheorem{example}[theorem]{Example}
 \newtheorem{exampleproposition}[theorem]{Example and Proposition}
 
\newcommand \Depth {\text{\rm Depth}}
\newcommand \mADM {m_\text{ADM}}
\newcommand \mH {m_\text{H}}
\newcommand \Hcal {\mathcal H}
\newcommand \Vol {\text{Vol}}
\newcommand{\be}{\begin{equation}}
\newcommand{\ee}{\end{equation}}
\newcommand \bel {\be\label} 
\newcommand \reg {\text{\rm reg}}
\newcommand \bei {\begin{itemize}}
\newcommand \eei {\end{itemize}}
\newcommand \weak {\text{\rm weak}}
\newcommand \Area {\text{\rm Area}}
\newcommand \loc {\text{\rm loc}} 
\newcommand \del \partial 
\newcommand \eps \epsilon 
\DeclareMathOperator{\RS}{RotSym} 
\newcommand \RSzero {\RS_m^{\weak, 0}}  
\newcommand \RSone {\RS_m^{\weak, 1}}  
\newcommand \RSnobar {\RS_m^{\weak}} 

\newcommand \RSonethree {\RS_3^{\weak, 1}}  

\newcommand \RSzerobar {\overline{\RS}{}_m^{\weak, 0}}  
\newcommand \RSonebar {\overline{\RS}{}_m^{\weak, 1}}  
\newcommand \RSbar {\overline{\RS}{}_m^{\weak}}  
\newcommand \RSzerobarthree {\overline{\RS}{}_3^{\weak, 0}}  
\newcommand \RSonebarthree {\overline{\RS}{}_3^{\weak, 1}}  

\let\oldmarginpar\marginpar
\renewcommand\marginpar[1]{\-\oldmarginpar[\raggedleft\footnotesize #1]%
{\raggedright\footnotesize #1}}
\newcommand{\N}{\mathbb{N}} 
\newcommand \R {\mathbb{R}} 
\newcommand{\E}{\mathbb{E}}
\newcommand{\Z}{\mathbb{Z}}
\newcommand{\diam}{\operatorname{diam}} 
\newcommand{\set}{\rm{set}}
\newcommand{\mass}{{\mathbf M}}
\newcommand{\intcurr}{{\mathbf I}}     

\newcommand{\vol}{\operatorname{Vol}}
\def \rmin   {r_{\min}}
\begin{document}

\title[The nonlinear stability of rotationally symmetric spaces]{The nonlinear stability of rotationally symmetric
\\
spaces with low regularity}

\author[Philippe G. L{\tiny e}Floch]{Philippe L{\scriptsize e}Floch$^1$}
\footnotetext[1]{Laboratoire Jacques-Louis Lions \& Centre National de la Recherche Scientifique, 
Universit\'e Pierre et Marie Curie, 4 Place Jussieu, 75252 Paris, France. Email: contact@philippelefloch.org}

\author[Christina Sormani]{Christina Sormani$^2$}
\footnotetext[2]{CUNY Graduate Center and Lehman College, 365 Fifth Avenue
New York, NY 10016-4309, USA. Email: sormanic@member.ams.org}

\date{}


\begin{abstract} We consider rotationally symmetric spaces with low regularity, which we regard as integral currents spaces or manifolds with Sobolev regularity and are assumed to have nonnegative scalar curvature. Relying on the flat distance and on Sobolev norms, we establish several nonlinear stability estimates about the ``distance'' between a rotationally symmetric manifold and the Euclidian space, which are stated in terms of the ADM mass of the manifold. Importantly, we make explicit the dependencies and scales involved in this problem, particularly the ADM mass, the depth, and the CMC reference hypersurface. Several notions of independent interest are introduced in the course of our analysis, including the notion of depth of a manifold and a scaled version of the flat-distance, the D-flat distance as we call it,  which involves the diameter of the manifold.  Finally we prove a compactness theorem for sequences of regions with uniformly bounded depth, whose outer boundaries have fixed area and an upper bound on Hawking mass. 
\end{abstract}

\maketitle 


\section{Introduction}
\label{sec:1}

It is of fundamental importance to understand the compactness of sequences of three dimensional 
asymptotically flat manifolds with nonnegative scalar curvature.   Recall that Schoen and Yau's positive mass theorem \cite{Schoen-Yau-positive-mass} establishes that the so-called ADM mass of such manifolds is nonnegative and vanishes if and only if the manifold is isometric to Euclidean space.
Naturally, the limits of such spaces will have low regularity, depending upon the notion of convergence used, and one still hopes to define nonnegative scalar curvature and notions like ADM and Hawking mass on such limit spaces.   Even the rotationally symmetric setting is not yet completely understood.  Lee and the second author \cite{LS1,LS2} have recently proven the stability of the positive mass theorem, in the sense that if a sequence of asymptotically flat, rotationally symmetric Riemannian manifolds, say $M_j$, with no closed interior minimal surfaces and nonnegative scalar curvature has ADM mass $m_{ADM}(M_j)\to 0$, then the sequence converges to Euclidean space in the intrinsic flat sense \cite{LS1}.  In \cite{LS2}, they showed that if a sequence of such $M_j$ approaches
equality in the Penrose Inequality then a subsequence converges in the intrinsic flat sense. However, these theorems strongly depend upon the fact that they were able to predict the limit space associated with these special sequences. More general sequences, in which only the ADM mass is bounded from above uniformly, can have limit spaces of very low regularity.  While the second author and Wenger in \cite{SW10,SW} have proven intrinsic flat limit spaces are always countably $\mathcal{H}^m$ rectifiable, the notion of nonnegative scalar curvature and Hawking mass on such spaces is difficult to define. 

On the other hand, the Einstein equations with solutions in the Sobolev space $H_\loc^1$ were extensively investigated by the first author together with Rendall \cite{LR} and Stewart \cite{LSt0,LSt}. This theory was motivated by a joint work with Mardare \cite{LM}, proving that a manifold with  $H^1_\loc$ regular metric admits an $L_\loc^2$ regular connection, whose curvature tensor is then defineable as a distribution. Thus, nonnegative scalar curvature and notions like Hawking mass which depend on mean curvature can be defined in a 
distributional sense.   Here, in the rotationally symmetric setting, we will be able to define nonnegative scalar curvature and Hawking mass and prove its monotonicity, under this $H_\loc^1$ regularity.  

Recall that the notion of $H_\loc^1$ regularity and $H_\loc^1$ convergence are
gauge dependent, in the sense that they depend upon a choice of coordinate charts, while intrinsic
flat convergence is defined using the metric geometry and does not depend upon gauge.   In this paper, we choose a specific gauge tied to the rotationally symmetric geometry and we are able to relate the two notions of convergence.    We also introduce the $D$-flat distance, a variation upon the intrinsic flat distance, which has good scaling properties and can be applied to sequences of regions $\Omega_j\subset M_j$ with a uniform upper bound on diameter $\diam(\Omega_j)\le D$.   

In particular, we study sequences of regions $\Omega_j\subset M_j$ within surfaces $\Sigma_j$
of {\em uniformly bounded depth} (a notion introduced here for the first time)  
\be \label{Depth-defn}
\Depth(\Sigma_j)=\sup\{ d_M(x, \Sigma_j):\,\, x\in \Omega_j\} \le D_0,
\ee
and uniformly bounded Hawking mass 
\be
m_H(\Sigma_j)\le M_0,
\ee
where 
\be \label{sigma-j}
\Sigma_j= \partial \Omega_j \setminus \partial M_j
\ee
 is a rotationally symmetric surface with fixed area 
\be 
\Area(\Sigma_j)=A_0
\ee
and
where the boundary $\partial M_j$ is either empty or a minimal surface.   Our spaces $M_j$ are assumed to be asymptotically flat, rotationally symmetric spaces with weak regularity admitting no closed interior minimal surfaces.

An outline of this paper is as follows. 
In Section~\ref{sec:2}, we introduce and study the various classes of spaces under consideration in this paper.
In Definition~\ref{def-RSweak} we extend the smooth class of Riemannian manifolds considered in \cite{LS1} and denoted by $\RS_m^\reg$, to classes
 $\RSone\subset \RSzero$ of $H^1_\loc$ and $L^2_\loc$ regularity, respectively.   We also introduce 
larger classes of the same low regularity but possibly with interior closed minimal CMC (constant mean curvature) hypersurfaces, denoted by 
 $\RSonebar\subset \RSzerobar$,  since such spaces
 may appear as limits.   We study the `profile functions' of these spaces, which are defined in (\ref{our-g}) below.  In Section~\ref{subsec:2.3}, we use these profile functions and define the mean curvature and scalar curvature in the distributional sense.  We also check the monotonicity of the Hawking mass in Proposition~\ref{thm-monotonicity} below.   
 
In Section~\ref{sec:3}, we prove that spaces $M\in \RSzero$ are countably $\mathcal{H}^m$ rectifiable metric spaces (and, for the convenience of the reader, we conclude here a brief review of this notion).   In Section~\ref{sec:4} we prove that tubular neighborhoods, $T_D(\Sigma)\subset M$ where $M\in \RSzero$ are integral current spaces (including a review of this notion).   This allows us to define the intrinsic flat distance between such regions. In Section~\ref{sec:5}, we review the notion of intrinsic flat distance and introduce the $D$-flat distance, which is first proposed in this paper; cf.~Definition~\ref{def-Dflat}. 
 
 In Section~\ref{sec:6}, we first review and then improve upon the stability of the positive mass theorem first proven by Lee and the first author \cite{LS1}.   We first rederive the original statement in \cite{LS1} by extending it to manifolds $M^m \in \RSone$; cf.~Theorem~\ref{thm-basic-LS}.
 We then reexamine the stability estimates in \cite{LS1} and establish {\em quantitative bounds} on the intrinsic flat distance, as well as on the $D$-flat distance and the difference in volumes between tubular neighborhoods $T_D(\Sigma)\subset M$ and annular regions in Euclidean space. 
These new estimates explicitly depend upon the parameters $m_{ADM}(M)$,
 $\Area(\Sigma)$ and $D$. (See Theorem~\ref{thm-new-F}).   
The technique of proof we propose here 
relies an arbitrary parameter which helps to "balance" contributions to the overall distance by selecting an optimal numerical value.   In Theorem~\ref{thm-adapted}, we thus provide precise bounds on the intrinsic flat distance, the $D$-flat distance and the difference in volumes between regions $U_D(\Sigma)$ which lie within $\Sigma$
and corresponding regions in Euclidean space, depending upon $m_H(\Sigma)$,
 $\Area(\Sigma)$, and $D$.   Next, in Theorem~\ref{thm-depth}, we provide such bounds for regions $\Omega$
of finite depth (in the sense \eqref{Depth-defn}) again depending upon the same parameters. 

In Section~\ref{sec:7}, we turn our attention to the Sobolev norms between the regions studied in Section~\ref{sec:6}.   We study thin regions in the $H^1$ norm using diffeomorphisms; cf.~Theorem~\ref{thm-sobolev-1}. Considering the possibility of very deep wells, we
realize that it is essential to study the {\em backwards profile functions} for level sets $\Sigma_0$ of given area.   These are defined in Definition~\ref{defn-back}.  In Theorem~\ref{thm-sobolev-no-diffeo}, we provide precise bounds on the $H^1[0,D]$ norm of the difference between backwards profile functions in $M$ and in Euclidean space, which depend upon the area $\Area(\Sigma_0)$, the Hawking mass $m_H(\Sigma_0)$, and $D$.   
 
In Section~\ref{sec:8} we prove our main compactness theorem which implies the following precompactness theorem. We refer to Theorems~\ref{compactness} and \ref{compSobolev} below for full statements. 

\begin{theorem}[Compactness framework in the intrinsic flat sense] 
Fix constants $A_0, D_0, M_0>0$. Consider a sequence of rotationally symmetric regions $\Omega_j\subset M_j$ lying within CMC spheres $\Sigma_j$ as stated in (\ref{sigma-j}), where $M_j$ have nonnegative scalar curvature and no interior minimal surfaces.   Assuming the uniform bounds 
\be
\aligned
Area(\Sigma_j) &=A_0,
\\
\Depth(\Sigma_j) &\le D_0,
\\
m_H(\Sigma_{j}) &\le M_0.
\endaligned
\ee
Then a subsequence (also denoted $M_j$) converges in the intrinsic flat sense to
a region $\Omega_\infty\subset M_\infty  \in \RSonebar$.   In particular, the limit space has and $H^1_\loc$ rotationally symmetric metric with nonnegative scalar curvature as defined in Section~\ref{sec:2}. 
By taking $\Sigma_\infty =\partial U_\infty \setminus \partial M_\infty \in M_\infty$, one has the following 
\be
\aligned
Area(\Sigma_\infty) & =A_0, 
\\
\Depth(\Sigma_\infty) & \le \liminf_{j \to +\infty} \Depth(\Sigma_j) \le D_0,
\\
m_H(\Sigma_{\infty}) &=\lim_{j \to +\infty} m_H(\Sigma_j) \le M_0,  
\endaligned
\ee
as well as 
\be
\vol(\Omega_\infty)=\lim_{j\to\infty}\vol(U_j)\le A_0D_0. 
\ee
(The relevant notions are defined as in Section~\ref{sec:2} below.) 
\end{theorem}

To establish this result, we first prove a Sobolev compactness theorem for the backwards profile functions and produce a candidate limit space in $\RSonebar$ (cf.~Theorem~\ref{compSobolev}). This convergence
is strong enough so that the limit space has nonnegative scalar curvature. We then apply a method by Lakzian and the first author  \cite{LakzianSormani} and transform the Sobolev convergence into intrinsic flat convergence; cf.~Proposition~\ref{compIF} below.   The convergence of the volume, area, and Hawking mass then follows from the convergence of the backwards profile functions
proven in Theorem~\ref{compSobolev}.  Intrinsic flat convergence alone is not strong enough to obtain convergence of these quantities. 

In Section~\ref{sec:9} we present several examples of particular interest. Example~\ref{no-scalar} demonstrates that while the notion of nonnegative scalar curvature is conserved in the limit, the scalar curvature does not converge.
Example~\ref{thin-well} (first presented in \cite{LS1}) has an increasingly thin well that disappears in the limit.   In  \cite{LS1}, this example was used to demonstrate why Gromov-Hausdorff convergence could not be used to prove the stability theorem.  Here, we use this example to demonstrate the importance of the backwards profile functions in Theorems~\ref{thm-sobolev-no-diffeo} and 
\ref{compSobolev}.  This example also demonstrates that the depth of a sequence need not converge.

One may naturally speculate on possible extensions of our theorems that do not require rotational symmetry.   It is of particular interest to understand the relationship between $H^1_\loc$ convergence and intrinsic flat convergence and whether one can rely on such relationship to also maintain nonnegative scalar curvature of the limit spaces without rotational symmetry.  One may also ask whether, under intrinsic flat or $H^1_\loc$ convergence, one can prove convergence of the Hawking mass (or another notion of quasilocal mass) for converging CMC hypersurfaces. 
 
\section*{Acknowledgements}
 
The authors gratefully acknowledge financial support from the National Science Foundation under Grant No.~0932078 000 via the Mathematical Science Research Institute, Berkeley, which the authors were visiting for the Fall Semester 2013. The first author (PLF) was also partially supported by the Agence Nationale de la Recherche through the grants ANR 2006--2-134423 and ANR SIMI-1-003-01. The second author (CS) was also partially supported by a PSC CUNY Research Grant and a National Science Foundation Grant DMS \#1309360. 


\section{Definition of rotationally symmetric spaces with low regularity}
\label{sec:2} 

\subsection{Definitions}

We begin with some definitions and properties about rotationally symmetric manifolds. We state first a definition for {\sl regular} manifolds. 

\begin{definition} 
\label{def:201}
The class $\RS_m^\reg$ of {\bf regular rotationally symmetric spaces} consists of $m$-dimensional, smooth topological manifolds with boundary, say 
$(M^m,g)$ endowed with a metric $g$ with $C^2$ regularity, which  
\bei 
\item are complete, rotationally symmetric, Riemannian manifolds such that the area of the distance sphere from the center tends to infinity when the distance approaches infinity,

\item admit no closed interior minimal hypersurfaces, and  either have no boundary or have a boundary which is a stable minimal hypersurface called an ``apparent horizon",

\item and have nonnegative scalar curvature. 

\eei  
\end{definition} 

For such manifolds, we can use geodesic coordinates and write
\bel{eq:201}
g = ds^2 + f(s) \, g_{S^{m-1}}, 
\ee
where $g_{S^{m-1}}$ is the standard unit metric on the $(m-1)$-sphere, $s$ is the distance from the boundary $\del M$, 
and the {\bf profile function} 
\be
f: [0, +\infty) \to [\rmin, +\infty)
\ee
 determines the overall geometry of the manifold.    Let
 \be\label{rmin}
 \rmin := f(0)=\lim_{s\to o} f(s)
 \ee
and we note that $f(0)=0$ if $M$ admits no boundary, while $f(0)>0$ if there is a boundary. 
Moreover, we say $M$ has a {\bf pole} (or a center) if $f(0)=0$ and thus $\partial M=\emptyset$. 
Finally, the orbits of the symmetry group are denoted by $\widetilde\Sigma_s$ and determine a CMC (constant mean curvature) foliation of the space. The profile function $f$ is strictly increasing  due to the restriction on the
non-existence of stable minimal surfaces.

A broad class of spaces is now obtained by relaxing the regularity requirement. 

\begin{definition} \label{def-RSweak}
The classes $\RSzero$ 
of  {\bf $L^2$ weakly regular rotationally symmetric spaces} 
consists of $m$-dimensional, smooth topological manifolds with boundary, say 
$(M^m,g)$, endowed with a metric with $L^2_\loc$, whose profile functions
$f \in L^2_\loc$ are strictly increasing from $r_{min}$ as in (\ref{rmin}).   
So that it satisfies all the properties listed in Definition~\ref{def:201} {\rm except} the last condition.

The class $\RSone$ of  {\bf $H^1$ weakly regular rotationally symmetric spaces} consists of $m$-dimensional, smooth topological manifolds with boundary, say 
$(M^m,g)$, endowed with a metric with $H^1_\loc$ regularity, which satisfy all the properties listed in Definition~\ref{def:201} in which the last condition is understood in the sense of distributions.

The classes and $\RSzerobar$ and $\RSonebar$ are defined similarly except\footnote{Our notation is motivated by a ``closure'' property established later in Section~\ref{sec:8}.} that one solely requires that the profile functions are non-decreasing and thus allows for interior minimal surfaces. 
\end{definition} 

The assumed $L^2_\loc$ ($H^1_\loc$, respectively) regularity means that, in any atlas of local coordinates, the metric coefficients belong to the space $L^2_\loc$ (resp.~$H^1_\loc$) of functions which (resp.~together with their first order derivatives) are locally square-integrable from the center (or pole). According to LeFloch and Mardare \cite{LM}, the connection of a manifold $(M^m,g) \in \RSonebar$ is well defined in the $L_\loc^2$ sense and its curvature tensors are well-defined as distributions.  The condition that the scalar curvature be nonnegative is thus understood here in the sense of distributions. Observe that no uniform regularity is assumed as one approaches the boundary of the manifold, which allows for a black hole in these spaces. 

Given $(M^m,g) \in \RSzerobar$, we introduce geodesic coordinates such that 
\bel{our-g}
g= ds^2 + f(s)^2 g_{S^{m-1}}, \qquad s \in (0, +\infty), 
\ee
where $g_{S^{m-1}}$ is the canonical metric on the unit $(m-1)$-dimensional sphere $S^{m-1}$. 
We observe that our definition yields the limited regularity 
\bel{our-f}
\aligned
& f \in L^2_\loc(0, +\infty) \quad \text{ if } (M^m,g) \in \RSzerobar, 
\\
& f \in H^1_\loc(0, +\infty) \quad \text{ if } (M^m,g) \in \RSonebar.  
\endaligned
\ee
In other words, the restriction of the profile function $f$ to any compact subset of $(0,\infty)$ is squared-integrable 
and, for the class $\RSonebar$, its first derivative in the distributional sense is also squared-integrable on that compact subset.


\subsection{Profile function and area of $\RSzerobar$ spaces}\label{subsec:2.2}

The local and global geometry of such manifolds $(M^m,g)$ will now be studied in terms of the properties of the profile function $f$. Several immediate but important observations are made in the rest of this section. We begin by discussing the regularity of the profile function $f$ and, until further notice, we assume that $(M^m,g) \in \RSzerobar$, so that the function $f$ is defined almost everywhere only. 
\bei 

\item  Our first assumption in Definition~\ref{def:201} about the area of the distance spheres tending to infinity when $s \to +\infty$ yields   
\bel{eq:210b}
\lim_{s\to+\infty} f(s)=+\infty.
\ee

\item The function $f$ is non-decreasing in $(0, +\infty)$ and thus 
\bel{eq:211}
f' \geq 0  \quad \text{ in the sense of distributions on } (0, +\infty)  
\ee
and the trace at the center $f(0):= \displaystyle\lim_{s \to 0 \atop  s >0} f(s)$ exists. 

\item Therefore, when the space has a pole, 
\be\label{f-pole}
f(1/k) \textrm{ approaches } 0 \textrm{ as } k\to +\infty,
\ee
while if $(M^m,g)$ does not have a pole then $f(s)\ge f(0)>0$ for all $s\ge 0$. 

\item In view of the monotonicity property of the function $f$, we can introduce its (right-continuous) pointwise representative by assigning a specific value at every $s \in (0, +\infty)$: 
\bel{f-values}
f(s) := \lim_{s_1 \to s \atop s_1 > s} f(s_1) = \liminf_{s_1 \to s} f(s_1).
\ee
Also, the function $f$ has countably many jump discontinuities. 

\item Finally, provided $(M^m,g)$ belongs to $\RSzero$, the condition \eqref{eq:210b} together with our assumption about the non-existence of closed interior minimal surfaces imply that $f$ has no local minima except possibly at the boundary $s=0$. 

\eei

Next, for each $s \in (0, +\infty)$, we consider the corresponding level set $\widetilde\Sigma_s$ of the distance function from the pole or the boundary, and we introduce the {\bf area function} $A=A(s)$ of these orbits of rotational symmetry, as well as their {\bf mean curvature} $H=H(s)$ given by 
\bel{eq:211b}
\aligned
A(s) =& \Vol(\widetilde\Sigma_s) = \omega_{m-1} (f(s))^{m-1} \qquad \text{at almost every } s >0, 
\\
H(s) =& (m-1) F'(s), \qquad \text{in the distributional sense,} 
\endaligned
\ee
where $\omega_{m-1}$ is a dimension-related constant and we have introduced the function 
\bel{eq:211c}
F(s) := \log f(s), \qquad \text{for almost all } s > 0.
\ee
These functions have only limited regularity, i.e.~thanks to \eqref{our-f} 
\bel{eq:212}
\aligned
& A \in L_\loc^2(0, +\infty) \quad \text{ when } (M^m,g) \in \RSzerobar, 
\endaligned
\ee
while the mean curvature $H$ is solely defined as a distribution.  Therefore, the mean curvature is not defined pointwisely, and the scalar curvature is not defined for {\sl all} slices $\widetilde\Sigma_s$ and, rather, we are working with a ``global'' definition dealing with the family of slices.

Another piece of notation will be useful. In view of \eqref{eq:211}, the area function $A: [0, +\infty) \to [A_{\min}, +\infty)$ is increasing (with $A_{\min} = A(0)$) and can be used to reparametrize the orbits of the symmetry group. So, for each $A_0 \in [A_{\min}, +\infty)$, we introduce the notation   
\bel{eq:212b}
\Sigma_{A_0} := \widetilde\Sigma_{s_0} \quad \text{ with $s_0$ characterized by } \Vol(\widetilde\Sigma_{s_0}) = A_0.  
\ee 

 
\subsection{Scalar curvature and Hawking mass of $\RSonebar$ spaces}
\label{subsec:2.3} 
 
In this section, we consider a space $(M^m,g) \in \RSonebar$. Then, the associated functions $A$ and $H$ have better regularity and, thanks to \eqref{our-f} 
and \eqref{eq:211b},  
\bel{eq:212cc}
\aligned
& A \in H_\loc^1(0, +\infty), 
\qquad
 H \in L^2_\loc(0, +\infty) \quad \text{ when } (M^m,g) \in \RSonebar.  
\endaligned
\ee
Importantly, the curvature of the space can now be defined, at least as a distribution. Specifically, for the scalar curvature, say $R=R(s)$, the expression originally derived for smooth metrics in \cite{LS1}
$$
R= {(m-1) \over (f(s))^2} \Big( (m-2) (1 - (f'(s))^2) - 2 f(s) f''(s)\Big) 
$$
does not immediately make sense since, in view of \eqref{our-f}, the second derivative $f''(s)$ is solely a distribution and is multiplied by the factors $(m-1)/ f(s)^{-2}$ and $f(s)$. It is convenient here to introduce the notation $F = \log f \in H^1_\loc(0, +\infty)$ and
we observe that $F''$ is defined as a distribution and the scalar curvature takes in the form
\bel{eq:213}
{R \over m-1} = -2 F'' - m {F'}^2 + (m-2) e^{-2F}. 
\ee
When the metric is sufficiently regular, this formula is equivalent to the standard formula for the scalar curvature, but \eqref{eq:213} now does make sense (as a distribution) even for metrics in our broad class $\RSzerobar$.
As expected from the general theory in \cite{LM}, we conclude that the scalar curvature 
\bel{eq:214}
R \, \text{ is well-defined as a distribution when } (M^m,g) \in \RSonebar.
\ee

Furthermore, our third assumption in Definition~\ref{def:201} that $R \geq 0$ in the distribution sense implies that $R$ is actually a locally bounded measure. In view of \eqref{eq:213}, this nonnegativity condition reads 
\bel{eq:214b}
F'' \leq - {m \over 2} {F'}^2 + {m-2\over 2} e^{-2F}, 
\ee
in which the left-hand side must understood in the sense of distributions but the right-hand side contains functions. So that our spaces enjoy the bounded variation regularity: 
\bel{eq:215}
f', F' \in BV_\loc(0, +\infty)
\ee
and, in particular, $f'$ is locally Lipschitz continuous and the condition \eqref{eq:211} becomes 
\bel{eq:2111}
\aligned
& f'(s) \geq 0 \qquad \text{ for all } s \in (0, +\infty)   && \text{ when } (M^m,g) \in \RSonebar, 
\\
& f'(s) > 0 \qquad \text{ for all } s \in (0, +\infty)   && \text{ when } (M^m,g) \in \RSone. 
\endaligned
\ee
Indeed, this inequality holds at all points, provided we introduce the right-continuous (say) pointwise representative of the function $f'$. 

We have some important consequence concerning the {\bf Hawking mass} $m_H=m_H(s)$, defined by
\bel{eq:220}
\aligned
2 \, m_H(s) &= (f(s))^{m-2} \big( 1 - (f'(s))^2 \big) 
\\
& = e^{(m-2) F(s)}  - e^{m F(s)}  (F'(s))^2. 
\endaligned
\ee
With some abuse of notation, we also use the radius $r=f(s)$ as an independent variable and we write $\mH=\mH(r)$. Furthermore, relying now on the monotonicity of the Hawking mass, we can introduce its limit at spatial infinity, denoted below by $\mADM \in [0, +\infty]$, which is nothing but the so-called ADM mass. In the following, we will assume that this limit is finite and seek for estimate in terms of this parameter. 

\begin{proposition}
\label{thm-monotonicity}
The Hawking mass \eqref{eq:220} of a manifold $(M^m,g) \in \RSonebar$ is a monotone non-decreasing\footnote{It is monotone inceasing if $(M^m,g) \in \RSone$.} function, which is bounded above and satisfies 
\bel{eq:297}
0 \leq \mH(\rmin)\le \mH(r)\le \mADM, 
\qquad 
r \in [\rmin, +\infty). 
\ee
If equality holds, that is, if the Hawking mass is a constant throughout the manifold, then $(M^m,g)$ is in fact regular, and coincides with 
Euclidean space (when $\mADM=0$) or the Riemannian Schwarzschild manifold of mass $\mADM>0$ given by 
\bel{eqn-def-Sch}
g=\left(1+\frac{2 \mADM}{ r^{m-2}-2 \, \mADM}\right)dr^2 + r^2 g_{S^{m-1}}.
\ee
\end{proposition}

From the positivity of the mass and \eqref{eq:2111}, we deduce the uniform bound
\bel{eq:233}
0 \leq  f'(s) \leq 1 \qquad \text{ for all } s \in (0, +\infty).
\ee 

\begin{proof} 
Namely,  in view of the inequality \eqref{eq:214b}, we have the monotonicity property 
\bel{eq:229}
m_H'(s) \geq 0
\ee
in the sense of distributions. On the other hand, by differentiating \eqref{eq:220} in the sense of distributions and using the chain rule for functions of bounded variation \cite{DLM}, we obtain
$$
\aligned
2 \, m_H'(s) 
& = (m-2) e^{(m-2) F(s)} F'(s)  - 2 e^{m F(s)} F'(s) F''(s) - m  e^{m F(s)} (F'(s))^3
\\
& = - 2 e^{m F(s)} F'(s) \Bigg( F''(s) + {m \over 2} (F'(s))^2 - {m-2 \over 2} e^{-2 F(s)}  \Bigg) \geq 0. 
\endaligned
$$
This calculation is justified, even at the level of weak regularity under consideration, provided one notices that the (ill-defined) product $F(s) F'(s)$ of a BV function by a measure is understood as a so-called Volpert's product; see for instance \cite{DLM}. Furthermore, our conditions in Definition~\ref{def:201} guarantees that $f'(0) = 0$ so that $\mH(\rmin)$ is nonnegative, so that the monotonicity of the Hawking mass yields \eqref{eq:297}. 
\end{proof} 

We complete this section with a remark and an example. 
 
\begin{remark}
For any manifold  $(M^m,g) \in \RSone$ the profile function $f=f(s)$ belongs not only to $H^1_\loc$ but in fact to $H^1$ (since $0 \leq f' \leq 1$ as a consequence of the Hawking mass bound).   
The function $f=f(s)$ need not belong to $H^2_\loc$, 
as seen in Example~\ref{not-H2}, below.
\end{remark}

\begin{example}\label{not-H2}
Let
\be
f(s) = 
\begin{cases}
a + b_1 s,          \qquad & s \in (0, s_1], 
\\
a + b_1 s_1 + b_2 (s - s_1), \qquad & s \geq s_1, 
\end{cases} 
\ee 
in which one chooses $s_1 >0$ and $a \geq 0$, as well as  $1 > b_1 > b_2$.   So the scalar curvature is positive, and $f$ is a profile function for a space $(M^m,g) \in \RSone$.    This function is not $H^2_\loc$.   In fact the second derivative $f''$ is bounded above but may approach $-\infty$, near the surface $s_1$.  
\end{example}


\subsection{Embedding of $\RSzerobar$ spaces in Euclidian space}
\label{sec:5}

\subsubsection*{The class of $\RSnobar$ spaces}

To proceed with the analysis of our classes of rotationally symmetric spaces, it is convenient to embed them first in Euclidian space. We provide here such a construction for any space $(M^m,g) \in \RSzero$. Indeed, our construction below requires nothing more than the conditions defining the broad class $\RSzero$. It will be important to precisely relate the regularity and the bounds in the variables $s$ and $r$, as we now do.  

Fix any $(M^m,g) \in \RSzero$. Since the function $f=f(s)=:r(s)$ is increasing and possibly discontinuous, it admits a non-decreasing and continuous inverse denoted by $s=s(r)$ for $r \in [\rmin, +\infty)$. The distributional derivative $s'(r) \geq 1$ is a locally bounded measure and we can introduce the {\bf height function} $z:   [\rmin, +\infty) \to [0, +\infty)$ by 
\bel{eq:510}
z(r) := \int_{\rmin}^r \sqrt{(s')^2 - 1},  \qquad r \in [0, +\infty), 
\ee
in which the integrant is actually a measure, defined (by Legendre transform, cf.~\cite{DemengelTemam}) as the composition of the measure $s'$ by the concave function $a \in [1, +\infty)\mapsto \sqrt{a^2 - 1}$. 
Observe that 
\bel{eq:511}
z=z(r) \quad 
\aligned
& \text{ is monotone non-decreasing}
\\
& \text{and continuous}
\endaligned
\quad 
\text{ if $(M^m,g) \in \RSzero$.} 
\ee

Observe that the function $z$ need not be increasing and may be constant on some intervals. In the class $\RSzero$, we have the following expressions in terms of the radial variable $r$: 
\be
A(r) = \omega_{m-1} r^{m-1}, \qquad 
H(r) = {m-1 \over r \sqrt{1 + (z')^2}},  
\qquad 
\mH(r) = {1 \over 2} r^{m-2} {(z')^2 \over \sqrt{1 + (z')^2}}.  
\ee
The function $A$ is of course smooth, but the mean curvature $H=H(r)$ (which was a measure in the variables $s$) is now a bounded function, at least away from the pole (if it exists).

Suppose next that $(M^m,g) \in \RSone$. We now have 
\bel{eq:511-third}
z=z(r) \quad 
\aligned
& \text{ is monotone increasing}
\\
& \text{and continuous}
\endaligned
\quad 
\text{ if $(M^m,g) \in \RSone$.} 
\ee
In terms of the function $z=z(r)$, the scalar curvature 
\be
{r R \over m-1} = - \Bigg( {1 \over 1+ {z'}^2} \Bigg)' + (m-2) {{z'}^2 \over 1 + {z'}^2} 
\ee
is now well-defined but solely as a distributions. 
The function $1/(1+ {z'}^2)$ therefore has locally bounded variation and, in particular, has countably many jumps.
Since $s'(r) \ge 1$ and the Hawking mass was shown to increase as $s$ increases, we see that the Hawking mass also increases as $r$ increases.
So, we conclude that the mass function  
\bel{Hawking-r-bound}
\mH(r) = {1 \over 2} r^{m-2} {(z')^2 \over {1 + (z')^2}}
\ee
is monotone increasing in $r$ and 
\bel{mH-ADM}
\lim_{r\to +\infty} \mH(r)=: m_{ADM}(M),
\ee
which we assume to be finite.

\subsubsection*{The class of $\RSbar$ spaces} Considering now a space in the broader class $(M^m,g) \in \RSzerobar$ and since 
 the function $f=f(s)=:r(s)$ is non-decreasing and possibly discontinuous, then its inverse $s=s(r)$ is also non-decreasing and possibly discontinuous. Then, the height function satisfies 
\bel{eq:511-bar}
z=z(r) \quad 
\aligned
& \text{ is monotone non-decreasing}
\\
& \text{and possibly discontinuous}
\endaligned
\quad 
\text{ if $(M^m,g) \in \RSzerobar$.} 
\ee
and 
\bel{eq:511-bar1}
z=z(r) \quad 
\aligned
& \text{ is monotone increasing}
\\
& \text{and possibly discontinuous}
\endaligned
\quad 
\text{ if $(M^m,g) \in \RSonebar$.} 
\ee


\section{Viewing $\RSzerobar$ spaces as countably rectifiable metric spaces}
\label{sec:3} 

\subsection{Countably rectifiable metric structure}

Our first objective is to eventually provide an interpretation of spaces in $\RSzerobar$ as integral current spaces (which we will need to estimate such spaces in the flat distance) but, first, in this section we show that such spaces can be viewed as rectifiable metric spaces. 

We denote by $\Hcal^m$ the $m$-dimensional Haussdorf measure. By definition, a metric space $(X,d)$ is said to be {\bf  countably $\Hcal^m$ rectifiable} if it admits a countable collection of bi-Lipschitz charts, say $\varphi_k: U_k \subset \R^m \to V_k \subset X$, where $U_k$ are Borel measurable sets and the family of sets $V_k$
cover almost all of $X$, in the sense that
$\Hcal_m\Bigg(X \setminus \bigcup_{k=1}^{+\infty} V_k \Bigg)=0$. 
For instance, any smooth Riemannian $m$-manifold $M$ with smooth Riemannian metric $g$ can be viewed as a countably $\Hcal^m$ rectifiable metric space, denoted by $(M, d_g)$, by setting 
\bel{d_g}
d_g(p,q) :=\inf\big\{ L_g(C): \,  C(0)=p, \, C(1)=q, \, C\,\, \textrm{piecewise smooth} \big\}
\ee
 for any two points $p, q \in M$, where the infimum is taken over all continuous and piecewise smoth curves with length defined by 
\bel{L_g}
L_g(C) :=\int_0^1 \big( g(C'(t), C'(t))\big)^{1/2} dt \in [0,+\infty].
\ee

We emphasize that the key property we will rely here is the monotonicity of the shape function $f$ describing the spaces in geodesic coordinates. 
In particular, our argument does not require the continuity of the metric.

\begin{proposition}[Viewing $\RSzerobar$ spaces as countably rectifiable metric spaces] 
\label{prop-rect}
A space $M^m \in \RSzerobar$ is a countably $\Hcal^m$ rectifiable metric space endowed with the distance $d_g$ defined
in (\ref{d_g})-(\ref{L_g}), provided the infimum is taken over piecewise smooth curves that avoid the pole (if it exists)
and, thus, in geodesic coordinates \eqref{our-g} with $C(t)=(s(t), \theta(t))$
(with $t \in [0,1]$, $s(t) \in (0, +\infty)$, and $\theta(t) \in S^{m-1}$)  
$$
L_g(C)=\int_0^1 \sqrt{|s'(t)|^2 \,+ \, | (f \circ s)(t) |^2\,|\theta'(t)|^2} \, dt, 
$$
where the precised (right-continuous) representative of the shape function $f$ is used in order to define the composite function 
$f \circ s$  (as in \eqref{f-values}). 
\end{proposition}


\subsection{Construction of the countably rectifiable structure}

Before we prove Proposition~\ref{prop-rect}, we need a few lemmas which will be used again elsewhere in the paper.    The first lemma is a standard lemma
from the study of smooth warped product spaces which we include since
it is not so well known although nowhere is it used that $f$ is smooth.

\begin{lemma}\label{lem-geod}
Let $(M, d)$ be defined as in Proposition~\ref{prop-rect}, let
$p_i=(s_i,\theta_i)\in M$ for $i=0,1$.   If $C(t)=(s(t), \theta(t))$
is piecewise smooth with $C(i)=p_i$ and $s(t)>0$ parametrized
so that $|\theta'(t)|=z$ almost everywhere where $z$ is constant,
and $C_2(t)=(s(t), \bar{\theta}(t))$, 
where 
$\bar{\theta}(t)$ is
a minimal geodesic in $S^{m-1}$ parametrized proportional to its arclength
with $\bar{\theta}(i)=\theta_i$ then $L(C_2) \le L(C)$.
\end{lemma}

\begin{proof}
First note that $z=L(\theta(0,1))$ viewed as a curve in the sphere
and that $|\bar{\theta}'(t)|=L(\bar{\theta}(0,1)) \le z$ since $\bar{\theta}$
is the minimal geodesic between the endpoints.   Then we have
$$
\aligned
L(C)&=\int_0^1 \sqrt{ |s'(t)|^2 \,+ \,|f(s(t))|^2\,|\theta'(t)|^2\,} \, dt
 =\int_0^1 \sqrt{ |s'(t)|^2 \,+ \,|f(s(t))|^2\,z^2\,} \, dt
\\
&\ge\int_0^1 \sqrt{ |s'(t)|^2 \,+ \,|f(s(t))|^2\,|\bar{\theta}'(t)|^2\,} \, dt=L(C_2). 
\endaligned 
$$
\end{proof}

The next lemma allows us to bound the distances between points from below. Recall that the shape function $f$ is strictly increasing and may have jump discontinuities, so that $f^{-1}$ is well-defined but is continuous and non-decreasing only. 

\begin{lemma} \label{lem-rect}
Let $(M, d)$ be defined as in Proposition~\ref{prop-rect}
and $k\in \mathbb{N}$.
Given a pair of points $p_i\in M$ such that $s(p_i)=1/k$ and
taking $\theta_i\in S^{m-1}$ to be the corresponding points in
the sphere $S^{m-1}$.  If
\be\label{lem-rect-1}
(f(1/k)-f(0))d_{S^{m-1}}(\theta_1,\theta_2)/2 < 2(1/k -s_k),
\ee
where $s_k:=f^{-1}( |f(1/k)-f(0)| /2 )$, 
then
\be
d_M(p_1, p_2)\ge |f(1/k)-f(0)|d_{S_2}(\theta_1,\theta_2)/2.
\ee
\end{lemma}

\begin{proof}
First observe that $s_k < 1/k$ and 
\be\label{s_k}
f(s) > |f(1/k)-f(0)|/2, \qquad \,s> s_k.  
\ee 
Now assume on the contrary that there is a piecewise
smooth curve $C$ (avoiding the pole or boundary joining
the points $p_i$) whose length has
\bel{lem-rect-2}
L(C)< |f(1/k)-f(0)|d_{S_2}(\theta_1,\theta_2)/2.
\ee
By Lemma~\ref{lem-geod} we can assume 
$C(t)=(s(t), \theta(t))$ where $\theta(t)$ is a minimal geodesic in the
sphere such that $|\theta'(t)|=d_{S_2}(\theta_1,\theta_2)$ for
almost every $t\in [0,1]$.   Thus
$$
\aligned
|f(1/k)-f(0)|/2
> L(C) /d_{S_2}(\theta_1,\theta_2) 
&= \int_0^1 |\theta'(t)|^{-1} \sqrt{ |s'(t)|^2 \,+ \,|f(s(t))|^2\,|\theta'(t)|^2\,} \, dt \\
&\ge \int_0^1 |f(s(t))|  \, dt \ge \min \big \{ f(s(t)): \, t\in [0,1] \big \}.
\endaligned
$$
Combining this with (\ref{s_k}) we see that there exists $t_0\in (0,1)$ such that $s(t_0)\le s_k$, thus
$$
\aligned
L(C)\ge \int_0^1  |s'(t)| \, dt
&  \ge |s(1)-s(t_0)|+|s(t_0)-s(1)|
\\
& \ge 2(1/k-s_k).
\endaligned
$$
Combining this with (\ref{lem-rect-1}) contradicts (\ref{lem-rect-2}) 
and we are done.
\end{proof}

\begin{proof}[Proof of Proposition~\ref{prop-rect}]
Let 
\bel{chart-1}
U_k= \{s\theta: \,s\in (1/k,k), \theta \in Q_k\} \subset \R^m,
\ee 
where $Q_k \subset S^{m-1}$ is a spherical cap of opening angle 
$\theta_k\in (0, \pi/4)$ chosen so that
\be
(f(1/k)-f(0))2\theta_k/2 < 2(1/k -s_k)
\ee
with $s_k$ defined as in Lemma~\ref{lem-rect} depending on $f$.

We define a countable collection of charts
\be
\varphi_k: U_k \subset \R^m \to V_k \subset M
\ee
where 
\bel{chart-4}
\varphi_k(s\theta)=(s, \theta).
\ee
Here we take all $k \in \mathbb{N}$
and, for each $k$, a finite collection of spherical caps $Q_k$ needed to cover $S^{m-1}$.   These charts cover all of $M$ except the pole or the boundary.

First we show $\varphi_k: V_k \to U_k$ are Lipschitz with
Lipschitz constant 
\be
L = \max\{1, \sqrt{2}k |f(k)|\}.
\ee 
Let $x_i=s_i\theta_i\in U_k\subset \R^m$ and join them by a line segment, $\gamma:[0,1]\to \R^m$
with $\gamma(i)=x_i$.   Since
$d_{S^{m-1}}(\theta_0, \theta_1)< \theta_k$ 
we can write
\be
\gamma(t)=s(t)\theta(t) \textrm{ where }
s(t)\in (\sqrt{2}/(2k), k) \textrm{ for all } t\in [0,1].
\ee 
Thus
$\varphi_k(x_i)$ are joined by the smooth curve $C(t)=(s(t), \theta(t))\in M$
and
$$
\aligned
d_g(\varphi_k(x_0), \varphi_k(x_1)
\le  L_g(C) 
& = \int_0^1 \sqrt{ |s'(t)|^2 \,+ \,|f(s(t))|^2\,|\theta'(t)|^2\,} \, dt \\
&\le  \int_0^1 \sqrt{ |s'(t)|^2 \,+ \,|f(k)|^2\,|\theta'(t)|^2\,} \, dt \\
&\le  L \int_0^1 \sqrt{ |s'(t)|^2 \,+ \, | \sqrt{2}/(2k) |^2   \,|\theta'(t)|^2\,} \, dt \le  L |x_0-x_1|.
\endaligned
$$
We claim that $\varphi_k^{-1}: V_k\to U_k$ is Lipschitz.   It
will take us three steps to prove this.

If $p,q\in V_k$, then for any $\eps>0$, there
exists a curve $C_1:[0,1]\to M$ from $p$ to $q$ such that
$L(C_1) \le d_g(p,q)+\eps$.   Applying Lemma~\ref{lem-geod},
we can write $C_1(t)=(s_1(t), \theta(t))$ and we can ensure that
$\theta(t)\in Q_k$ since $Q_k$ is convex in $S^{m-1}$.   

We define a new curve
$C_2:[0,1]\to V_k$ by $C_2(t)=(s_2(t), \theta(t))$ and 
$s_2(t)=\max\{s_1(t),j\}$, 
so that $C_2(0)=C_1(0)=p$ and $C_2(1)=C_1(1)=q$.   Furthermore
\bel{eq:epsil}
L(C_2) \le L(C_1)\le d_g(p,q)+\eps, 
\ee
since $f$ is monotone increasing (which is the key assumption required in our construction). 

If $s_2(t) \ge 1/k$, let $C_3(t)=C_2(t)$ for all $t\in [0,1]$.    Otherwise let
$t_1=\inf\{t\in [0,1]: \, s_2(t)< 1/k\}$ and $t_2=\sup\{t\in [0,1]: \, s_2(t)< 1/k\}$,
and set 
$C_3(t)=C_2(t)$ for all $t\in [0,1]\setminus (t_1,t_2)$
 and for
$t\in (t_1, t_2)$ let $C_3(t)=(1/k, \theta_3(t))$ where $\theta_3(t)$
is running minimally from $\theta(t_1)$ to $\theta(t_2)$.  Since
$d_{S^2}( \theta(t_1),\theta(t_2))< 2\theta_k$
and
$$
(f(1/k)-f(0))d_{S^2}( \theta(t_1),\theta(t_2)/2 < 2(1/k -s_k), 
$$
we have (by Lemma~\ref{lem-rect})
$$
L(C_2(t_1,t_2)) \ge (f(1/k)-f(0))d_{S_2}(\theta(t_1),\theta(t_2))/2.
$$
Thus, we find
$$
\aligned
L(C_3(t_1, t_2))
&= \int_{t_1}^{t_2} f(1/k) |\theta_3'(t)|\, dt \\
&=  f(1/k) d_{S_2}(\theta(t_1),\theta(t_2))
\le  \frac{2f(1/k)}{f(1/k)-f(0)} L(C_2(t_1,t_2))
\endaligned
$$
and
$$
\aligned
L(C_3)&=L(C_3(0,t_1))+L(C_3(t_1,t_2))+L(C_3(t_2,1)) \\
&\le L(C_2(0,t_1))+\frac{2f(1/k)}{f(1/k)-f(0)}L(C_2(t_1,t_2))+L(C_2(t_2,1)) \\
&\le \frac{2f(1/k)}{f(1/k)-f(0)} L(C_2). 
\endaligned
$$

Next, since $C_3(t)=(s_3(t), \theta_3(t)) \subset V_k$, we
can define a curve 
\be
\varphi_k^{-1}\circ C_3(t)=s_3(t)\theta_3(t)\in U_k \subset \R^m
\ee
 running from 
$\varphi_k^{-1}(p)$ to $\varphi_k^{-1}(q)$ whose length can be
estimated as follows
$$
\aligned
\left|\varphi_k^{-1}(p)- \varphi_k^{-1}(q)\right| 
&\le \int_0^1 \sqrt{s_3'(t)^2 + (s_3(t))^2 \theta'(t)^2 }\, dt \\
&\le   \left( 1+\frac{1}{f(1/k)}\right) \int_0^1 \sqrt{s_3'(t)^2 + (f(s_3(t)))^2 \theta'(t)^2 }\, dt
\endaligned
$$
since $f(s_3(t))\ge f(1/k)$, so that 
$$
\aligned 
\left|\varphi_k^{-1}(p)- \varphi_k^{-1}(q)\right| 
&=   \left( 1+\frac{1}{f(1/k)}\right) \int_0^1 g(C_3'(t), C_3'(t))^{1/2} \, dt 
= \left( 1+\frac{1}{f(1/k)}\right)\, L_g(C_3) \\
&\le  \left( 1+\frac{1}{f(1/k)}\right)\frac{2f(1/k)}{f(1/k)-f(0)}  \, L_g(C_2) \\
&\le   \left( 1+\frac{1}{f(1/k)}\right)\frac{2f(1/k)}{f(1/k)-f(0)} \, (d_g(p,q)+\eps).
\endaligned
$$
Thus $\varphi_k^{-1}: V_k\to U_k$ is Lipschitz with Lipschitz constant
$ \left( 1+\frac{1}{f(1/k)}\right)\frac{2f(1/k)}{f(1/k)-f(0)}$.

We now have a countable collection of bi-Lipschitz charts which cover all of
$M$ except the pole or the boundary.   The pole clearly has Hausdorff measure $0$ since it is a single point.  The boundary also has 
$\Hcal^m$ measure $0$ since it is a sphere of radius $f(0)$ and
dimension $m-1$.  
\end{proof}


\section{Viewing $\RSzerobar$ spaces as integral current spaces}
\label{sec:4} 

\subsection{Background on integral currents}

Federer and Flemming \cite{FF,Federer} introduced the notion of an integral current in Euclidean space as a way to generalize the notion of a smooth oriented submanifold with boundary.   If $\psi: M^m \to \R^N$ be a bi-Lipschitz embedding of a smooth oriented submanifold, it can
be viewed as an $m$-dimensional integral current, $T=\psi_\#[[M]]$, which acts on differential $m$-forms, $\omega$, so that
$T(\omega) = \int_M \psi^*\omega$.  
In this way they were able to define the weak convergence of submanifolds
viewed as currents, $T_j \to T$ if and only if $T_j(\omega)\to T(\omega)$ for all differential forms of compact support.  They proved that this weak convergence is equivalent
to flat convergence when the sequence has a uniform bound $\vol(M)+\vol(\partial M)$.  The limits of the submanifolds under this notion of convergence are called integral currents.   These integral currents, $T$, are {\bf rectifiable} in the sense that there exists a countable collection bi-Lipschitz charts 
$\psi_k: U_k \to V_k \subset \R^N$ such that
$
T(\omega)=\sum_k \int_M h_k \psi_k^*\omega, 
$
where $h_k \in \Z$.   Furthermore, one can define a weighted volume, called the {\bf mass}: 
\bel{FFmass}
\mass(T) =\sum_k |h_k| \vol(\psi_k(U_k))<+\infty.
\ee
In addition they have a boundary defined by $\partial T(\omega)=T(d\omega)$ and this boundary is also an integral current.   In particular $\mass(\partial T)<+\infty$.

Ambrosio and Kirchheim extended the notion to integral currents on complete metric spaces $(Z,d)$ by taking them to act on tuples of Lipschitz functions, $(f, \pi_1,...,\pi_m)$ rather than smooth forms.   If $\psi: M^m \to Z$   is Lipschitz
then $T=\psi_\#[[M]]$ is defined so that
\be
T(f, \pi_1,...,\pi_m)=\int_M (f\circ \psi) \, d(\pi_1\circ\psi)\wedge \cdots \wedge d(\pi_m\circ\psi).
\ee
More generally an $m$-dimensional rectifable currents, $T$, defined on $m+1$ tuples of Lipschitz
functions $(f, \pi_1,...,\pi_m)$ is defined by a collection of bi-Lipschitz charts
$\varphi_k:U_k \to V_k\subset Z$ such that
\bel{gen-T}
T(f, \pi_1,....\pi_m):= \sum_{k=1}^{+\infty} \int_{U_k} h_k f\circ \varphi_k
\, d(\pi_1\circ\varphi_k) \wedge \cdots d(\pi_m\circ \varphi_k),
\ee
where $h_k$ are positive integers and the $U_k$ are Borel measurable sets in $\R^m$.    They also define mass, $\mass(T)$, which we will refer to as Ambrosio-Kirchheim mass, which they require to be finite.  This mass does not satisfy (\ref{FFmass}) but it can be bounded:
\bel{AKmass}
\mass(T) \le C_m\sum_k |h_k| \Hcal_m (\varphi_k(U_k))<+\infty, 
\ee
where $C_m$ is a constant depending on the dimension.
 A rectifiable current $T$ is called an {\bf integral current} (written $T\in \intcurr_m(Z)$)
if $\partial T$ has finite mass where
\be
\partial T(f, \pi_1,...,\pi_{m-1})=T(1, \pi_1,...,\pi_{m-1}), 
\ee
in which case they prove $\partial T$ is also rectifiable.
They define weak convergence of integral currents testing against the tuples of functions which agrees with flat convergence when the
$\mass(T)+\mass(\partial T)$ is uniformly bounded from above.   
They also define $\set(T)\subset Z$ as the set of positive density of
$T$ and prove that this is a countably $\Hcal^m$ rectifiable set using the
same charts as the ones in (\ref{gen-T}).

Finally, given a Lipschitz map $\varphi: Z_1 \to Z_2$, and an integral
current $T \in \intcurr_m(Z_1)$, they define the push-forward $\varphi_{\#}T \in \intcurr_m(Z_2)$ as follows
\bel{push-forward}
\varphi_{\#}T (f, \pi_1,...,\pi_m)=T(f\circ \varphi, \pi_1\circ \varphi,...,\pi_m\circ\varphi.
\ee
When $\varphi$ is metric isometric embedding, that is
\be\label{metric-isom-embed}
d_{Z_2}(\varphi(x), \varphi(y))=d_{Z_1}(x,y), \qquad  x,y\in Z_1,
\ee
then one has
\be\label{mass=}
\mass(\varphi_{\#}T)=\mass(T).
\ee


\subsection{Background on integral current spaces} 

In this paper we are not studying submanifolds of any metric space, but rather sequences of Riemannian manifolds.  In Sormani and Wenger~\cite{SW}, the notion of an integral current space was introduced as a way to generalize the notion of a smooth oriented Riemannian manifold with boundary.  The intrinsic flat distance
between integral current spaces was defined to extend the notion of Federer-Flemming's flat distance between integral currents in Euclidean space.   Thus one is able to take intrinsic flat limits of Riemannian manifolds and study their limits which are metric spaces called integral current spaces.   One may
also consider sequences of integral current spaces when one does not wish to require the full regularity required to define a smooth Riemannian manifold
with a smooth metric tensor.

An {\bf integral current space} $(X,d,T)$ is a weighted oriented countably $H^m$ rectifiable metric space, $X$, endowed with an integral current structure $T\in \intcurr_m(\bar{X})$ such that $X=\set(T)$.   This means that $X$ has a countable collection of bi-Lipschitz charts, $\varphi_k: U_k \subset \R^m \to V_k \subset X$ where
$U_k$ are Borel measurable sets and where $V_k$
cover almost all of $X$:
\be\label{almost-all}
H_m\left(X \setminus \bigcup_{k=1}^{+\infty} V_k \right)=0
\ee
and an $m$-dimensional integral current structure, $T$, defined on $m+1$ tuples of Lipschitz
functions $(f, \pi_1,...,\pi_m)$ as follows:
\bel{gen-T-2}
T(f, \pi_1,....\pi_m):= \sum_{k=1}^{+\infty} \int_{U_k} h_k f\circ \varphi_k
\, d(\pi_1\circ\varphi_k) \wedge \cdots d(\pi_m\circ \varphi_k), 
\ee
where $h_k$ are positive integers.
In addition $T$ must have finite Ambrosio-Kirchheim mass,
$\mass(T)< +\infty$, and the boundary current,
\be
\partial T(f, \pi_1,...,\pi_{m-1}):=T(1, f, \pi_1,...,\pi_{m-1}), 
\ee
which is $m-1$ dimensional
must have finite Ambrosio-Kirchheim mass, $\mass(\partial T)<+\infty$.   
  
In \cite{SW} it was shown that any compact oriented smooth Riemannian manifold with boundary $(M^m,g)$ can be considered to be an integral current space $(M, d, T)$, by setting the metric $d=d_g$ as in
(\ref{d_g})-(\ref{L_g})
and taking
\bel{T}
T(f, \pi_1,...\pi_m)=\int_M f \, d\pi_1\wedge \cdots d\pi_m.
\ee
One can easily find a collection of oriented bi-Lipschitz charts with disjoint images that cover almost all of $M$ as in (\ref{almost-all}).
Taking $h_k=1$ we can define $T$ as in (\ref{gen-T-2})
with $h_k=1$ to obtain (\ref{T}).  The Ambrosio-Kirccheim mass of $T$ is then just the
volume of $M$, that is, $\mass(T)=\vol_m(M)$, which is finite as required.   The {\bf boundary} of $T$, is defined as in the work of Ambrosio and Kirchheim as
\be
\aligned
\partial T(f, \pi_1,...\pi_{m-1})
 & =T(1,f,\pi_1,...,\pi_{m-1})= \int_{M} 1 \, df \wedge d\pi_1\wedge \cdots d\pi_{m-1}
\\
& =
\int_{\partial M} f \, d\pi_1\wedge \cdots d\pi_{m-1}
\endaligned
\ee
also has finite mass, $\mass(\partial T)=\vol_{m-1}(\partial M)$.

Note that if a smooth Riemannian manifold $M$ is non-compact
and asymptotically flat, then its volume is infinite and so it is
not an integral current space.  However smooth 
compact subregions of $M$ are integral current spaces.
For example, Lee and Sormani \cite{LS1} applied the
fact that tubular neighborhoods of symmetric spheres, $\Sigma$,
\be 
T_R(\Sigma)= \big\{x: d(x, \Sigma)\le R \big\}
\ee
are integral current spaces.  Thus we could study how close they were
in the intrinsic flat sense to the corresponding regions in Euclidean space.
In the next section, we show that we can similarly study 
tubular neighborhoods of symmetric spheres within 
$M\subset \RSzerobar$.


\subsection{Tubular neighborhoods viewed  as integral current spaces}

Here, we prove Propositions~\ref{prop-int-curr-space} and~\ref{prop-int-curr-space-2} by showing that tubular neighborhoods and inner tubular neighborhoods in $M\subset \RSzerobar$  are integral current spaces.
The manifolds themselves have infinite volume and are this not integral current spaces.

\begin{proposition}[Tubular neighborhoods viewed as integral current spaces]
 \label{prop-int-curr-space}
Let $M^m\in \RSzerobar$ and $\Sigma=s^{-1}(s_0)$ 
be a level set of the associated function $s$.  Fix any $D>0$ and define the distance $d_g$
as in (\ref{d_g})-(\ref{L_g}).  Then, the tubular neighborhood
\be 
T_R(\Sigma) := \big\{ x: d_g(x, \Sigma)\le D \big\}
\ee
is an
integral current space when viewed as a metric space with the restricted metric $d_g$ and  whose current structure is defined by (\ref{T}).   In addition,
the boundary of the tubular neighborhood viewed as an integral current spaces is the boundary of the tubular neighborhood viewed as a submanifold where integral current structure is defined as usual with opposing orientations on the outer and inner boundaries
\be\label{eqn-int-curr-space}
\partial T(f, \pi_1,...,\pi_{m-1})
=\int_{s^{-1}(s_0+D)} f d\pi_1\wedge \cdots \wedge d\pi_{m-1}
-\int_{s^{-1}(s_D)} f d\pi_1\wedge \cdots \wedge d\pi_{m-1}, 
\ee
where $s_D=\max\big\{ s_0-D,0 \big\}$.
\end{proposition}

Note that the definition of the current structure does not depend on
the metric $g$.   However, in order to prove that this is indeed an integral
current space, we must show that $T$ is an integral current: that there
is a collection of bi-Lipschitz 
charts $\varphi_k: U_k \subset \R^m \to V_k \subset \bar{X}$ where
$U_k$ are Borel measurable sets and where $V_k$
cover almost all of $\bar{X}$ satisfying $\ref{gen-T-2}$ with finite mass and
that the boundary also has finite mass.  The definition of Ambrosio-Kirchheim mass and of bi-Lipschitz depends upon $d_g$.

\begin{proof} 
Let $k_0 \in \N$ be chosen so that $k_0> s_0+D$ and $1/k_0<s_0$.   
For $k=k_0$ let
$
a_{k_0}= \max\{1/{k_0}, s_0-D\} \textrm{ and } b_{k_0}=s_0+D.
$
and for $k>k_0$ let
$
a_k= \max\{1/k, s_0-D\} \textrm{ and } b_k=a_{k-1}
$
so that $(a_k, b_k)$ are pairwise disjoint and so that the closure of their union is $[s_0-R, s_0+R]$.  Let
$
k_{max}=\sup\{ k:\,\, a_k<b_k\} \in [k_0, +\infty].
$
Observe that $k_{max}<+\infty$ unless there is a pole.   
When
there is a pole we will use the fact that $f(0)=0$ and
(\ref{f-pole}) to control the infinite series that we will need to deal with.

Recall that, in the proof of Proposition~\ref{prop-rect} in (\ref{chart-1})-(\ref{chart-4}), we found a countable bi-Lipschitz collection of charts covering almost all of $M$.   We now choose 
\bel{chart-5}
U_{k,\alpha}= \left\{s\theta: \,s\in (a_k,b_k), \,\theta \in Q'_{k,\alpha}\,\right\} \subset 
U_{k} \subset \R^m,
\ee 
where $Q_{k,\alpha}$ are triangular disjoint subsets of the spherical caps  $Q_k \subset \mathbb{S}^{m-1}$ such that
$
\bigcup_{\alpha=1}^{N_k} Q_{k,\alpha} = \mathbb{S}^{m-1}.
$
Setting 
$
\varphi_{k,\alpha}(s\theta)=\varphi_k(s, \theta)
$
as in (\ref{chart-4}) and setting
$
V_{k,\alpha}=\varphi_{k,\alpha}(U_{k, \alpha})\subset V_k \subset M, 
$
we have bi-Lipschitz charts
$$
\varphi_{k,\alpha}: U_{k,\alpha} \subset \R^m \to V_{k,\alpha} \subset M
$$
with disjoint images such that
$$
\bigcup_{k=k_0}^{k_{max}} \bigcup_{\alpha=1}^{N_k} V_{k,\alpha}=T_D(\Sigma)\subset M.
$$
So in particular this tubular neighborhood is a countable $\mathcal{H}^m$ rectifiable set.   

We next verify that the $T$ defined in (\ref{T}) is a rectifiable current:
\be\label{sum-for-T}
\aligned
T(h, \pi_1,...\pi_m)&=\int_M h \, d\pi_1\wedge \cdots d\pi_m
= \sum_{k=k_0}^{k_{max}} \sum_{\alpha=1}^{N_k} \int_{V_{k,\alpha}}
h \, d\pi_1\wedge \cdots d\pi_m\\
&= \sum_{k=k_0}^{k_{max}} \sum_{\alpha=1}^{N_k} \int_{U_{k,\alpha}}
(h\circ\varphi_{k,\alpha}) \, d(\pi_1\circ\varphi_{k,\alpha})\wedge \cdots 
d(\pi_m\circ\varphi_{k,\alpha}).
\endaligned
\ee
Thus when $k_{max}<+\infty$ we are done.

When $k_{max}=+\infty$ we claim that
\be\label{claim-converging}
\sum_{\alpha=1}^{N_k} \int_{U_{k,\alpha}}
(h\circ\varphi_{k,\alpha}) \, d(\pi_1\circ\varphi_{k,\alpha})\wedge \cdots 
d(\pi_m\circ\varphi_{k,\alpha})
\le C_k \left(\frac{1}{k(k-1)}\right),
\ee
where for $k$ sufficiently large
$$
C_k\le\sup\{|h|\} Lip(\pi_1)\cdots Lip(\pi_m)
\omega_{m-1} (f(s_0+D)^{m-1}
$$
and so we have a converging sum in (\ref{sum-for-T}). Thus $T$ is a rectifiable current in
this case as well.

To prove our claim first observe that
$$
\aligned
\int_{U_{k,\alpha}}
(h\circ\varphi_{k,\alpha}) \, d(\pi_1\circ\varphi_{k,\alpha})\wedge \cdots 
d(\pi_m\circ\varphi_{k,\alpha})
&= \int_{V_{k,\alpha}}
(h) \, d(\pi_1)\wedge \cdots d(\pi_m) \\
&\le \sup\{|h|\} Lip(\pi_1)\cdots Lip(\pi_m) \mathcal{H}_m(V_{k,\alpha}).
\endaligned
$$
Note also that for $k>k_0$, by the monotonicity
of $f$  
 we have
$$
\aligned
\sum_{\alpha=1}^{N_k}\mathcal{H}_m(V_{k,\alpha})
=\vol\left(s^{-1}(a_k,b_k)\right)
&\le  \omega_{m-1} (f(b_k))^{m-1} (b_k-a_k)\\ 
&\le \omega_{m-1} (f(b_k))^{m-1} \left(\frac{1}{k-1}-\frac{1}{k}\right)
\\
&
\le \omega_{m-1} (f(s_0+D)^{m-1} \left(\frac{1}{k(k-1)}\right).
\endaligned
$$

In fact, we have 
$$
\aligned
\mass(T)
& \le C_m \sum_{k=k_0}^{k_{max}} \sum_{\alpha=1}^{N_k} \mathcal{H}_m(\varphi_{k,\alpha}(U_{k,\alpha})) 
\\
& \le C_m \mathcal{H}_m(T_D(\Sigma))\le \omega_{m-1}(2D)(f(s_0+D))^{m-1} <+\infty.
\endaligned
$$
To establish that $T$ is an integral current, we now check that the
boundary to $T$ is a rectifiable current.   Observe that
\be
\label{bndry-sum} 
\aligned
\partial T(h, \pi_1,...,\pi_{m-1}) &= T(1,h,\pi_1,...,\pi_{m-1})\\
&=\sum_{k=k_0}^{k_{max}} \sum_{\alpha=1}^{N_k} \int_{U_{k,\alpha}}
d(h\circ\varphi_{k,\alpha}) \wedge d(\pi_1\circ\varphi_{k,\alpha})\wedge \cdots 
d(\pi_{m-1}\circ\varphi_{k,\alpha})\\
&=\sum_{k=k_0}^{k_{max}}  \int_{a_k}^{b_k} \int_{S^{m-1}}
d(h\circ\varphi_{k}) \wedge d(\pi_1\circ\varphi_{k})\wedge \cdots 
d(\pi_{m-1}\circ\varphi_{k})
=\sum_{k=k_0}^{k_{max}}  \,\,\,\, B_k-A_k,   
\endaligned
\ee
with 
$$
\aligned
B_k&=\int_{\{b_k\}\times S^{m-1}}
(h\circ\varphi_{k})\,\,d(\pi_1\circ\varphi_{k})\wedge \cdots 
d(\pi_{m-1}\circ\varphi_{k}),
 \\
A_k&=
\int_{\{a_k\}\times S^{m-1}}  
(h\circ\varphi_{k}) \,\, d(\pi_1\circ\varphi_{k})\wedge \cdots 
d(\pi_{m-1}\circ\varphi_{k}).
\endaligned
$$
When $k_{max}<+\infty$ this suffices to show that $\partial T$ is
rectifiable.

When $k_{max}=+\infty$ we must show the sum in (\ref{bndry-sum}) is finite.   
To do this, we adapt the standard proof that an alternating series converges when its terms converge to $0$.  Recall that $k_{max}=+\infty$ only if $M$ has a pole. 
By $(\ref{f-pole})$,
we know that there exists a sequence $\eps_j \to 0$ such that
\be
f(s) \le f(\eps_j) \le 1/j^2, \qquad  s\le \eps_j.
\ee
Choose a sequence $k_0=k_0$, $k_j>k_{j-1}$ such that
$
b_{k_j}<\eps_j.
$
Thus, we have 
\be
\sum_{j=1}^{+\infty} \omega_{m-1} f(b_{k_j})^{m-1} <+\infty.
\ee
Since $b_k=a_{k-1}$ for $k>k_0$, we have $B_k=A_{k-1}$, and 
so 
$
B_{k_j} - A_{k_j}= \sum_{k=k_{j-1}}^{k_j} ( B_k-A_k ).
$
Thus, we find 
$$
\partial T(h, \pi_1,...,\pi_{m-1}) =\sum_{j=0}^{+\infty}  \,\, B_{k_j}-A_{k_j}.
$$
This series is absolutely converging, since 
$$
\aligned
\sum_{j=0}^{+\infty}  \,\, |B_{k_j}|+|A_{k_j}|
\le \sum_{j=0}^{+\infty}  \,\, 2|B_{k_j}| 
& \le  \sum_{j=1}^{+\infty} \left|\int_{\{b_{k_j}\}\times S^{m-1}}
(h\circ\varphi_{k_j})\,\,d(\pi_1\circ\varphi_{k_j})\wedge \cdots 
d(\pi_{m-1}\circ\varphi_{k_j}) \right|\\
& \le  \sum_{j=1}^{+\infty} \left|\int_{\varphi_{k_j}(\{b_{k_j}\}\times S^{m-1})}
(h)\,d(\pi_1)\wedge \cdots d(\pi_{m-1}) \right|,
\endaligned
$$
thus 
$$
\aligned
\sum_{j=0}^{+\infty}  \,\, |B_{k_j}|+|A_{k_j}| 
& \le  \sum_{j=1}^{+\infty} 
\sup{|h|} Lip(\pi_1)\cdots Lip(\pi_{m-1}) \mathcal{H}_{m-1}(\varphi_{k_j}(\{b_{k_j}\}\times S^{m-1})\\
& \le  \sum_{j=1}^{+\infty} 
\sup{|h|} Lip(\pi_1)\cdots Lip(\pi_{m-1}) \omega_{m-1}(f(b_{k_j}))^{m-1}\\
& \le  \sum_{j=1}^{+\infty} 
\sup{|h|} Lip(\pi_1)\cdots Lip(\pi_{m-1}) \omega_{m-1}(1/j^2)^{m-1}<+\infty.
\endaligned
$$
Thus $\partial T$ is rectifiable and so $T$ is an integral current.

We may now use the fact that $b_k=a_{k-1}$ and telescope the
possibly infinite sum to see that
$$
\aligned
\partial T(h, \pi_1,...,\pi_{m-1}) =& 
\int_{\{b_{k_0}\}\times S^{m-1}}
(h\circ\varphi_{k}) \wedge d(\pi_1\circ\varphi_{k})\wedge \cdots 
d(\pi_{m-1}\circ\varphi_{k}) \\
&- 
\int_{\{a_{k_{max}}\}\times S^{m-1}}
(h\circ\varphi_{k}) \wedge d(\pi_1\circ\varphi_{k})\wedge \cdots 
d(\pi_{m-1}\circ\varphi_{k}), 
\endaligned
$$
where $a_{k_{max}}=0$ if $k_{max}=+\infty$.  So $a_{k_{max}}=D_k$.   Thus we obtain
(\ref{eqn-int-curr-space}).
\end{proof}

The next statement is established by following exactly the lines of the proof of Proposition~\ref{prop-int-curr-space} 
(except that $b_{k_0}=s_0$). 

\begin{proposition} \label{prop-int-curr-space-2}
Let $M^m\in \RSzerobar$ and $\Sigma=s^{-1}(s_0)$ be a level set of the function $s$.   Fix $D>0$ and define the distance $d_g$
as in (\ref{d_g})-(\ref{L_g}).  Then, the inner tubular neighborhood
\be 
U_D(\Sigma)=s^{-1}([s_0-D, s_0])
\ee
is an
integral current space when viewed as a metric space with the restricted metric $d_g$ and  whose current structure is defined by (\ref{T}).   In addition,
the boundary of the tubular neighborhood viewed as an integral current spaces is the boundary of the tubular neighborhood viewed as a submanifold where integral current structure is defined as usual with opposing orientations on the outer and inner boundaries
\be\label{eqn-int-curr-space-2}
\partial T(f, \pi_1,...,\pi_{m-1})
=\int_{s^{-1}(s_0)} f d\pi_1\wedge \cdots \wedge d\pi_{m-1}
-\int_{s^{-1}(s_D)} f d\pi_1\wedge \cdots \wedge d\pi_{m-1}, 
\ee
where $s_D=\max\big\{ s_0-D,0 \big\}$.
\end{proposition}


\section{The intrinsic flat distance and the D-flat distance}
\label{sec:5bis} 

\subsection{Reviewing the intrinsic flat distance}

The intrinsic flat distance between two oriented Riemannian manifolds with
boundary of finite volume (or more generally a pair of integral current spaces)
was introduced in Sormani and Wenger \cite{SW}.  This notion is gauge invariant.  

Given $M_i=(X_i, d_i, T_i)$ of the same dimension, $m$,
we recall that the intrinsic flat distance,
\be
d_{\mathcal{F}}(M_1, M_2) = \inf\left\{ d_F^Z(\varphi_{1\#}T_1, \varphi_{2\#}T_2):
\,\, \varphi_i :M_i \to Z\, \right\}
\ee
where the infimum is taken over all complete metric spaces, $Z$, and over all metric isometric embeddings 
$\varphi_i: X_i \to Z$:
\be\label{IF-1}
d_Z(\varphi_i(x), \varphi_i(y)= d_{X_i}(x,y), \qquad  x,y\in Z.
\ee
Here the flat distance in $Z$, 
\bel{IF-2}
d_F^Z(\varphi_{1\#}T_1, \varphi_{2\#}T_2)=\inf\left\{ \mass(A) + \mass(B):\,\,  A+\partial B=\varphi_{1\#}T_1-\varphi_{1\#}T_2\right\}
\ee
where the infimum is taken over all $A \in \intcurr_m(Z)$ and 
$B\in \intcurr_{m+1}(Z)$ such that $ A+\partial B=\varphi_{1\#}T_1-\varphi_{1\#}T_2$.   The notion of a flat distance for integral currents in Euclidean space was introduced by Federer and Flemming and applied to solve the Plateau Problem at least in a weak sense \cite{FF}.   

The intrinsic flat distance is a distance and is gauge invariant in the sense that
given two precompact integral current spaces, $M_i$,
\be
d_{\mathcal{F}}(M_1, M_2)=0
\ee
if and only if 
there is a current preserving isometry
\be
\psi: X_1 \to X_2 \textrm{ such that } \psi_{\#}T_1=T_2.
\ee
In particular if $M_1$ is a Riemannian manifold then $\psi$ is an orientation preserving isometry.   

\begin{remark}\rm
\label{rmrk-est}
If $M_i^m$ are Riemannian manifolds and one can find oriented metric isometric embeddings $\varphi_i$ from $U_i=M_i \setminus A_i\subset M_i$ into the boundary of a common Lipschitz Riemannian manifold $B^{m+1}$, such that 
\be
\int_{\varphi_1(U_1)} \omega-
\int_{\varphi_2(U_2)} \omega=
\int_{B} d\omega + \int_{A_3}\omega
\ee
for some $A_3\in \partial B$.  Then one can construct a common metric space
$Z$ by gluing $M_i$ to $B$ along the images of $\varphi_i(U_I)$, and set
$A_i=M_i\setminus U_i$.  After verifying that $\varphi_i$ extend to
metric isometric embeddings $\varphi_i: M_i \to Z$, one can then bound
the intrinsic flat distance as follows:
\be
d_{\mathcal{F}}(M_1, M_2) \le \vol(B^{m+1}) +
\vol(A_1^m) + \vol(A_2^m) +\vol (A_3^m).
\ee
This is the construction used by Lee and Sormani \cite{LS1} to prove tubular neighborhoods in rotationally symmetric manifolds around CMC surfaces of fixed area $\alpha_0$ with increasingly small ADM mass converge in the intrinsic flat sense to tubular neighborhoods in Euclidean space.   We will use this technique here as well.
\end{remark}

Naturally there is a notion of pointed intrinsic flat convergence: a sequence of oriented Riemannian manifolds with boundary, $M^m_j$, with basepoints $p_j \in M_j$ converges in the pointed intrinsic flat sense to a Riemannian manifold $M^m_\infty$ with basepoint $p_\infty\in M_\infty$ if and only if for almost every $D>0$ the balls $B_{p_i}(D)$ converge in the intrinsic flat sense to $B_{p_\infty}(D)$:
\be
\lim_{i\to +\infty} d_{\mathcal{F}}(B_{p_i}(D), B_{p_\infty}(D)) =0.
\ee
In \cite{LS1} sequences of rotationally symmetric manifolds whose
ADM mass is decreasing to $0$ are shown to converge in the pointed intrinsic flat sense to Euclidean space if the points are selected to lie on CMC surfaces of fixed area, $\alpha_0$.   Naturally it would mean nothing if the points were allowed to diverge to infinity since the spaces are asymptotically flat.  The theorem is false if the points are taken to be the poles as they can descend down deeper and deeper wells.    So it was of critical importance to fix the location of the points in some invariant way.


\subsection{Introducing the D-flat distance}

The intrinsic flat distance does not scale when the pair of Riemannian manifolds are rescaled since it is a sum of two terms of different dimension.   It has this property since it is based upon Federer and Flemming's flat norm in Euclidean space which is a norm with respect to rescaling the weight of the currents rather than rescaling the space they sit in. Recall that Lee and Sormani~\cite{LS1} had suggested studying the scalable flat distance which scales like length:
\be
d_{\mathcal{F}}(M_1, M_2) = \inf\left\{ \mass(A)^{1/m}+ \mass(B)^{1/(m+1)}: \,\, \varphi_i :M_i \to Z\,,\,\, A+\partial B=\varphi_{1\#}T_1-\varphi_{2\#}T_2\right\}
\ee
where the infimum is taken over all $Z$ and $\varphi_i$ as in 
(\ref{IF-1}) and over all $A, B$ as in (\ref{IF-2}).

In the present paper, we introduce the following new notion.

\begin{definition} \label{def-Dflat}
The {\bf D-flat distance} between pairs of Riemannian manifolds with the same upper bound, $D$, on their diameter:
\be
d_{D\mathcal{F}}(M_1, M_2) = \inf\left\{ \mass(A)+ \frac{\mass(B)}{D}: \,\, \varphi_i :M_i \to Z\,,\,\, A+\partial B=\varphi_{1\#}T_1-\varphi_{1\#}T_2\right\}, 
\ee
where the infimum is taken over all $Z$ and $\varphi_i$ as in 
(\ref{IF-1}) and over all $A, B$ as in (\ref{IF-2}).   
\end{definition}

One may also try other notions of convergence dividing by volume or   
by diameter in different ways.   Based upon our study of sequences of spaces in $\RSonebar$ with bounded ADM mass, the definition above seems to be the most natural notion.  We refer to our application of this notion in the following sections. 

It is immediate (and quite natural) to define the {\bf pointed D-flat convergence}
for any sequence of Riemannian manifolds without assuming an upper bound on diameter.  We just require that for almost every $D>0$
\be
\lim_{i\to +\infty} d_{D\mathcal{F}}(B_{p_i}(D), B_{p_\infty}(D)) =0.
\ee
Furthermore, it is clear that Sormani-Wenger's compactness theorem remains true for our distance. 


\section{Nonlinear stability in the intrinsic flat distance}
\label{sec:6}

\subsection{Reviewing the $\mathcal{F}$-stability estimate} 
\label{subsect-review-LS}

Throughout this section, we restrict attention to the class of spaces $M^m\in\RSone$  whose ADM mass is finite.  
Hence, we are thus restricting attention to(with strictly increasing profile functions and to spaces without interior minimal surfaces. 
We observe first that the theorem established by Lee and Sormani \cite{LS1} for {\sl regular} manifolds immediately extends to this weak class.   However, \cite{LS1} did not establish quantitative and compactness estimates, which is our main objective in the present paper. 
Recall that $\E^m$ denotes the Euclidean space of dimension $m$. 

\begin{theorem}[$\mathcal{F}$-stability estimate]
\label{thm-basic-LS}   
Given any $\eps,D, A_0>0$ and an interger $m\in \N$
there exists a constant $\delta=\delta(\eps, D, A_0, m)>0$
such that, for every space $M^m\in\RSone$ with $\mADM(M)<\delta$, 
\bel{eq:first}
 d_{\mathcal{F}}\big( T_D(\Sigma_0)\subset M^m,  T_D(\Sigma_0)\subset \E^m\big) 
< \eps. 
 \ee
 where  $\Sigma_0$ is the symmetric sphere of area  
$\vol_{m-1}(\Sigma_0)=A_0$, and $T_D(\Sigma)$ is the 
tubular neighborhood of radius $D$ around $\Sigma_0$.   
 \end{theorem}
 
 It should be noted that $T_D(\Sigma_0)\subset M^m$ and
 $T_D(\Sigma_0)\subset \E^3$ need not be diffeomorphic in
 order to achieve this closeness in the intrinsic flat sense.

\begin{proof} Here, we explain briefly why the statement holds on our weaker class of spaces and we also record the key estimates that will be useful later in the paper. This result was proven by applying the technique described in Remark~\ref{rmrk-est} defining a Lipschitz continuous, Riemannian manifold $B=B_1\cup B_2$ where $B_1$ is defined by the embedding into $\E^{m+1}$:
$$
\aligned
B_1 &=\big\{\,\big(\,x_1,...,x_m, z(r(x_1,...,x_m))\,\big)\,\,: \,\, r(x_1,....,x_m)\in (r_\eps, r_{D^+}) \,\,\big\} \subset \E^{m+1},
\\
B_2 &=U_1\times [0, S_M]
\endaligned
$$
and $U_1$ is a strip defined with a precise choice of $S_M>0$,
$$
\aligned
U_1 &= r^{-1}(r_\eps, r_{D^+})\, \subset\, T_{D}(\Sigma_{\alpha_0}), 
\qquad \qquad 
r_{D^+} =\max\big\{r(p):\, p \in T_{D}(\Sigma_{\alpha_0})\big\}. 
\endaligned
$$
Here, the radius 
$r_\eps \ge r_{D^-}=\min\big\{r(p):\, p \in T_{D}(\Sigma_{\alpha_0})\big\}$
was carefully chosen in \cite{LS1} so that 
$A_1 := T_{D}(\Sigma_{\alpha_0})\setminus U_1$ has sufficiently small volume $\vol(A_1)$.

We set $U_2=r^{-1}(r_\eps, r_{D+})\subset \E^m$
so that
$$
T_D(\Sigma_{\alpha_0})=A_{2,1} \cup A_{2,2} \cup U_2 \subset \E^m, 
$$
where
$
A_{2,1}=A_2=r^{-1}(r_{D-}, r_\eps) \subset \E^m
$
is possibly empty and
$
A_{2,2}=A_0=r^{-1}(r_{D^+}, r_0+D) \subset \E^m, 
$
with $\alpha_0=\omega_{m-1}r_0^{m-1}$.
Finally, the region 
$
A_3=A_{3,1}\cup A_{3,2} \cup A_{3,3} \subset \partial B
$
has
$$
\aligned
A_{3,1}&=\mathbb{S}_{r_{D^+}}\times [0, S_M] \subset \partial B_2,
\\
A_{3,2}&=\mathbb{S}_{r_\eps}\times [0, S_M] \subset \partial B_2,
\\
A_{3,3}&=\mathbb{S}_{r_{D^+}}\times [z(r_\eps), z(r_{D_+})] \subset \partial B_1, 
\endaligned
$$
where $A_{3,2}$ is possibly empty. (See Figure 3 in \cite{LS1}.)    

We have proven earlier that we can also isometrically embed our
Riemannian manifold $(M^m, g) \in \RSone$ into $\E^{m+1}$ using the height function $z$ which is known to be continuous.   
By
(\ref{Hawking-r-bound}) we have 
\bel{Hawking-r-bound-2}
\mH(r) = {1 \over 2} r^{m-2} {(z')^2 \over \sqrt{1 + (z')^2}}\le m_{ADM}(M), 
\ee
which is exactly as in \cite{LS1}.   We can choose the same strip width $S_M$
as in \cite{LS1} and the same $r_\eps$ and achieve the exact same
theorem as in \cite{LS1} only now for a sequence of manifolds in 
$\RSone$ whose ADM mass approaches $0$.   This completes the proof of 
Theorem~\ref{thm-basic-LS}.
\end{proof}


\subsection{Re-visiting the $\mathcal{F}$ stability estimate}
\label{subsect-reexamine}
 
From now and for simplicity in the presentation and without genuine loss of generality, we focus on $3$-dimensional spaces. 
In the present work, we examine the estimate \eqref{eq:first} more carefully so as to get a {\sl quantitative estimate} on the flat distance between 
$T_D(\Sigma_{\alpha_0})\subset M^3$ and $T_D(\Sigma_{\alpha_0})\subset \E^3$.   We begin by recalling certain constants from \cite{LS1}, especially 
\be\label{delta}
\delta :=m_H(r_{D^+}) \le m_{ADM}(M).
\ee
In Lemma 4.2 in \cite{LS1}, let us choose $\delta$ small depending upon an earlier choice of $r_\eps< r_0$ so that 
\be\label{z'Q}
z'(r) \le Q,\qquad  r\ge r_\eps, 
\ee
giving a specific formula for $Q$ depending on $\delta$ and $r_\eps$:
\be
Q=\sqrt{\frac{2\delta}{(r_\eps-2\delta)}}> \sqrt{\frac{2\delta}{(r_0-2\delta)}}.
\ee
Observe that $Q$ is scale invariant. Here we would prefer not to pick $r_\eps$ before we choose $\delta$ since we are not examining a sequence with $\delta_i\le m_{ADM}(M_i)\to 0$.   Instead we solve for 
$$
r_\eps=(2\delta(1+Q^{-2})) < r_0, 
$$
so that (\ref{z'Q}) is a consequence of the choice of $r_\eps$.

We now write the estimates from \cite{LS1} for $\vol(B)$ and $\vol(A)$ as {\sl functions of the parameters} $Q$ and $\delta$, $D$ and $\alpha_0$.  In the next section we
will choose the {\sl optimal value} for $Q$ and obtain a new and stronger estimate on the
intrinsic flat as well as D-flat distances. Examining the proof of Lemma 4.1 in \cite{LS1} we see that
$$
\aligned
\vol(A_1)
&\le 4\pi r_\eps^{2}D 
 \le \omega_2 (2\delta(1+Q^{-2}))^2 D
\endaligned
$$
and 
$$
\aligned
\vol(A_2)= (4/3)\pi r_\eps^3 
&\le (4/3)\pi r_\eps^2 r_0
\\
&\le (4/3)\pi (2\delta(1+Q^{-2}))^2 D.
\endaligned
$$

Since $z'(r)\le Q$, Lemma 4.3 in \cite{LS1} shows that
$$
\vol(A_0)\le D Q 4\pi (r_0+D)^2.
$$
Also one can estimate
$$
\aligned
\vol(A_{3,3}) 
&\le 4\pi(r_{D^+})^2 (z(r_{D^+})-z(r_\eps))
\\
&\le 4\pi(r_0+D)^2 Q (r_{D^+}-r_\eps)\le 4\pi(r_0+D)^2 Q (2D), 
\endaligned
$$
since $r_\eps>r_{D+}-2D$.

Lemma 4.5 in \cite{LS1} chooses the strip width
$
S_M=\sqrt{C(2D+\pi r_0+C)}, 
$
where
$
C=(4D+2\pi r_0)Q
$
to guarantee the metric isometric embedding of $U_1$ into $B$.
Requiring now 
\be
Q\le 1/2
\ee
so that $C\le 2D+\pi r_0$ and
$$
\aligned
S_M&=\sqrt{(4D+2\pi r_0)Q(2D+\pi r_0+2D+\pi r_0)}\\
&\le  2(D+\pi r_0) \sqrt{Q},
\endaligned
$$
we arrive at 
$$
\aligned
\vol(A_{3,1})
&= S_M 4\pi (r_0+D)^2 = 8\pi(r_0+D)^2(D+\pi r_0) \sqrt{Q},
\\
\vol(A_{3,2})&=S_M 4\pi r_\eps^2 \le  8\pi(r_0+D)^2(D+\pi r_0) \sqrt{Q}. 
\endaligned
$$
Summing over all of these we get
$$
\aligned
\vol(A_3)&= \vol(A_{3,1})+\vol(A_{3,2})+ \vol(A_{3,3})\\
&\le 16 \pi (r_0+D)^2(D+\pi r_0) \sqrt{Q} + 4\pi(r_0+D)^2 Q (2D)\\
&\le 4\pi(r_0+D)^2(6D+4\pi r_0)\sqrt{Q}, 
\endaligned
$$
since $Q\le \sqrt{Q}$, and thus 
thus
\be
\label{A1}
\aligned
\vol(A)&=\vol(A_0)+\vol(A_1)+\vol(A_2)+ \vol(A_3)\\
&\le D Q 4\pi (r_0+D)^2 + \omega_2 (2\delta(1+Q^{-2}))^2 D\\
& \quad + (4/3)\pi (2\delta(1+Q^{-2}))^2 D +4\pi(r_0+D)^2(6D+4\pi r_0)\sqrt{Q}\\
&\le4\pi (8D+4\pi r_0)\left(\delta^2(1+Q^{-2})^2 + (r_0+D)^2\sqrt{Q} \right).
\endaligned
\ee

We can estimate $\vol(B)$ next, as follows: 
$$
\aligned
\vol(B_1)
&=  \int_{r_\eps}^{r_{D^+}} 4\pi r^2 (z(r) - z(r_\eps) ) \, dr
\le  \int_{r_\eps}^{r_{D^+}}4\pi r^2 \int_{r_\eps}^r z'(s) \, ds\,dr \\
&\le  \int_{r_\eps}^{r_{D^+}}4\pi r^2 \int_{r_\eps}^r Q \, ds\,dr 
\le  \int_{r_\eps}^{r_{D^+}} 4\pi r^2 Q (r-r_\eps)\,dr, 
\endaligned
$$
thus 
$$
\aligned
\vol(B_1)
&\le  \int_{r_\eps}^{r_{D^+}} 4\pi (r_{D^+})^2 Q (2D) \,dr \\
&\le  4\pi (r_0+D)^2 Q(2D)(r_{D^+}-r_\eps) \\
&\le  4\pi (r_0+D)^2 Q(2D)(2D)\le 8\pi D (r_0+D)^2 \sqrt{Q}(2D). 
\endaligned
$$
We also estimate 
$$
\aligned
\vol(B_2)=
 S_M \vol(U_2)
&=  2(D+\pi r_0) \sqrt{Q} 
\int_{r_\eps}^{r_{D^+}} 4\pi r^2 \sqrt{1+z'(r)^2} \, dr\\
&\le  2(D+\pi r_0) \sqrt{Q} 
\int_{r_\eps}^{r_{D^+}} 4\pi r^2 \sqrt{1+Q^2} \, dr,
\endaligned
$$
thus $$
\aligned
\vol(B_2)=
&\le  2(D+\pi r_0) \sqrt{Q} \sqrt{1+Q^2} (4/3)\pi(r_{D^+}^3-r_\eps^3)\\ 
&\le  2 (D+\pi r_0) \sqrt{Q} \sqrt{1+Q^2} 4\pi (r_0+D)^2(2D)\\ 
&\le  8\pi D (r_0+D)^2(D+\pi r_0) \sqrt{Q} \sqrt{2}.
\endaligned
$$
Thus, we obtain 
\be\label{B1}
\aligned
\vol(B)&=\vol(B_1)+\vol(B_2)\\
&\le8\pi D (r_0+D)^2(4D+2\pi r_0) \sqrt{Q}. 
\endaligned
\ee

Note also  that we have estimates on
\be
\label{V1}
\aligned
\vol(T_D(\Sigma_{\alpha_0})\subset \E^3) 
&\le \vol(T_D(\Sigma_{\alpha_0})\subset M),
 \\
 \vol(T_D(\Sigma_{\alpha_0})\subset M)
 &= \vol(A_1) + \vol(U_2)\\
&\leq 4\pi (2\delta(1+Q^{-2}))^2 D  +\sqrt{1+Q^2} (4/3)\pi(r_{D^+}^3-r_\eps^3) \\ 
&\le  4\pi (2\delta)^2(1+Q^{-2})^2 D  +(1+Q) \vol(T_D(\Sigma_{\alpha_0})\subset \E^3),
  \\ 
\vol(\partial T_D(\Sigma_{\alpha_0})\subset\E^3) 
&\le 4\pi r_\eps^2 + 4\pi(r_{D^+})^2\\
&\le  4\pi  (2\delta)^2(1+Q^{-2}))^2 +4\pi(r_0+D)^2.
\endaligned
\ee


\subsection{A new estimate in the intrinsic flat distance}

We may now prove the following theorem which strengthens the results in \cite{LS1} and justifies our introduction of the D flat distance.   Note also how the sum of the D flat distance and the
difference in volumes have the same dependence on $\delta$.   

\begin{theorem}[Quantitative estimate in the intrinsic flat distance]
\label{thm-new-F}   
Suppose $(M^3, g) \in \RSonethree$ and $m_{ADM}(M^3)=\delta$ with
\be
\delta \le \min\left\{\frac{r_0}{32}, \frac{8^5 r_0^9}{(r_0+D)^8}\right\} 
\ee
and let $\Sigma_{\alpha_0}$ be the CMC surface of area
$
\alpha_0=4\pi r_0^2
$
then one has 
\be
\aligned
 d_{\mathcal{F}}(\,T_D(\Sigma_0)\subset M^3\,,\, T_D(\Sigma_0)\subset \E^3\,)
& < (1+D) \eps(D, r_0, \delta), 
\\
 d_{D\mathcal{F}}(\,T_D(\Sigma_0)\subset M^3\,,\, T_D(\Sigma_0)\subset \E^3\,)
& <  2\eps(D, r_0, \delta),
\endaligned
 \ee
  where
$\eps(D, r_0,\delta):=
 48\pi (2D+\pi r_0)(r_0+D)^{16/9}\delta^{2/9}$
and, furthermore,  
 \be
\aligned
\left| \vol(T_D(\Sigma_{\alpha_0}\subset M^3))- \vol(T_D(\Sigma_{\alpha_0}\subset \E^3))\right|
&\le \eps(D, r_0,\delta),
\\
 \vol(\partial T_D(\Sigma_{\alpha_0}\subset M^3)) 
& \le
 \eps(D,r_0,\delta)/(8D+4\pi r_0) + 4\pi(r_0+D)^2.
\endaligned
 \ee
 \end{theorem}
 
 It should be noted that $T_D(\Sigma_0)\subset M^m$ and
 $T_D(\Sigma_0)\subset \E^3$ need not be diffeomorphic in
 order to achieve this closeness property in the intrinsic flat sense.
 
 \begin{proof}
 We first choose the best $Q$ subject to the constraints that
 $Q\le 1/2$ and 
$
Q> \sqrt{\frac{2\delta}{(r_0-2\delta)}}
$
to minimize
$$
\vol(A)=4\pi (8D+4\pi r_0)\left(\delta^2(1+Q^{-2})^2 + (r_0+D)^2\sqrt{Q} \right).
$$
 Taking $q=\sqrt{Q}$ and observing that $(1+Q^{-2})^2 \le 8q^{-8}$ so that
$$
\vol(A)\le F(q):= 4\pi (8D+4\pi r_0)\left(8\delta^2 q^{-8} + (r_0+D)^2q \right),
$$
we find 
$$
0=F'(q)=4\pi (8D+4\pi r_0)\left(-64\delta^2 q^{-9}  + (r_0+D)^2 \right).
$$
So the critical point is
$
q=\left( \frac{64\delta^2} {(r_0+D)^2}\right)^{1/9} 
$
and the best choice for 
\be
Q= \left( \frac{8\delta} {(r_0+D)}\right)^{4/9}
\ee
if it fits the constraints and, by the hypothesis $\delta\le r_0/32$, 
$$
Q\le \left( \frac{r_0/4} {(r_0+D)}\right)^{4/9}\le \left( 1/4 \right)^{4/9}\le 1/2.
$$
Again, by the hypothesis of the theorem, we have 
\be\label{delta-bound}
\delta\le\frac{8^5 r_0^9}{(r_0+D)^8}.
\ee
Given (\ref{delta-bound}), we find 
$$
\aligned
\sqrt{\frac{2\delta}{(r_0-2\delta)}}
&>\frac{\sqrt{2}\delta^{1/2}}{r_0^{1/2}}
=\frac{8^{3/18}\delta^{8/18} \delta^{1/18}}{r_0^{1/2}}\\
&\le \frac{8^{8/18}\delta^{4/9}}{(r_0+D)^{4/9}} \frac{8^{4/9}(r_0+D)^{4/9}\delta^{1/18}}{8^{5/18}r_0^{1/2}}
\le \left(\frac{8\delta}{(r_0+D)}\right)^{4/9},
\endaligned
$$
so $Q$ fits the constraints.
We now substitute our choice for $Q$ into 
$$
\aligned
\vol(A)&\le4\pi (8D+4\pi r_0)\left(\delta^2(1+Q^{-2})^2 + (r_0+D)^2\sqrt{Q} \right)\\
&=4\pi (8D+4\pi r_0)\left(\delta^2 8Q^{-4} + (r_0+D)^2\sqrt{Q} \right)\\
&\le4\pi (8D+4\pi r_0)\left(\delta^2 8 \left( \frac{8\delta} {(r_0+D)}\right)^{-16/9} + (r_0+D)^2 \left( \frac{8\delta} {(r_0+D)}\right)^{2/9} \right)\\
&\le 4\pi (8D+4\pi r_0)\left( (1/8)^{7/9} (r_0+D)^{16/9}\delta^{2/9} + 8^{2/9}(r_0+D)^{16/9}\delta^{2/9}\right), 
\endaligned
$$
thus, with our notation, 
\be\label{A2}
\vol(A)
\le  \eps(D, r_0,\delta).
\ee
Combining this with
(\ref{A1}) and (\ref{B1}), we see that
$$
\max\{\vol(B)/D, \vol(A)\} \le  \eps(D, \alpha_0, m_{ADM}(M)), 
$$
which gives our estimate on the intrinsic flat and D-flat distances.
Rearranging (\ref{V1})
and substituting our choice for $Q$ we obtain
$$
\aligned
\left|\vol(T_D(\Sigma_{\alpha_0})\subset M) -\vol(T_D(\Sigma_{\alpha_0})\subset \E^3) \right|
&\le 4\pi (2\delta)^2(1+Q^{-2})^2 D + Q \vol(T_D(\Sigma_{\alpha_0})\subset \E^3)  \\
&\le 8\pi \delta^2(8Q^{-4}) D + \sqrt{Q} (4/3)\pi (r_0+D)^3 \\
&\le 16\pi (r_0+D)\left(\delta^2(8Q^{-4})  + \sqrt{Q} \pi (r_0+D)^2\right) 
\le\eps(D, r_0,\delta).
\endaligned
$$
Finally we have
\be\label{V8}
\aligned
\vol(\partial T_D(\Sigma_{\alpha_0})\subset\E^3) 
&\le  4\pi  (2\delta)^2(1+Q^{-2}))^2 +4\pi(r_0+D)^2 \\
&\le  \eps(D,r_0,\delta) /(8D+4\pi r_0)+ 4\pi(r_0+D)^2.
\endaligned
\ee
\end{proof}


\subsection{Nonlinear stability of inner regions}

Let $U_D(\Sigma)$ is the  part of the tubular neighborhood of radius $D$ around $\Sigma_0$
that lies within $\Sigma_0$.   

\begin{theorem}[Nonlinear stability of inner regions]
\label{thm-adapted}    
Suppose $(M^3, g) \in \RSonethree$ and $m_{H}(\Sigma_{\alpha_0}) =: \delta$ with
$\delta \le r_0/32$, 
where $\Sigma_{0}$ be the CMC surface of area
$\alpha_0=4\pi r_0^2$
with
\be
\aligned
 d_{\mathcal{F}}(\,U_D(\Sigma_0)\subset M^3\,,\, U_D(\Sigma_0)\subset \E^3\,)
& < (1+D)\eps_U(\delta, D, r_0),
\\
 d_{D\mathcal{F}}(\,U_D(\Sigma_0)\subset M^3\,,\, U_D(\Sigma_0)\subset \E^3\,)
& < 2\eps_U(\delta, D, r_0),  
\endaligned
 \ee
 where $\eps_U(D, r_0,\delta)=48\pi (2D+\pi r_0)r_0^{16/9}\delta^{2/9}$
and, furthermore, one has 
 \be\label{V2}
\aligned
\left| \vol(U_D(\Sigma_{\alpha_0}\subset M^3))- \vol(U_D(\Sigma_{\alpha_0}\subset \E^3))\right|
&\le \eps_U(D, r_0,\delta),
 \\
 \left| \vol(\partial U_D(\Sigma_{\alpha_0}\subset M^3))- \vol(\partial U_D(\Sigma_{\alpha_0}\subset \E^3))\right|
& \le
 2\eps(D,r_0,\delta)/(8D+4\pi r_0).
\endaligned
 \ee
 \end{theorem}

\begin{proof}
To see this proof we return to Section~\ref{subsect-review-LS}
and observe that we should take $r_{D^+}=r_0$ when defining the
regions $A$ and $B$.   Then in Section~\ref{subsect-reexamine},
everywhere that we estimates $r_{D^+}\le r_0+D$, we have
$r_{D^+}=r_0$.   So instead of (\ref{A1}) we have
\be
\label{A2-2}
\vol(A)
\le4\pi (8D+4\pi r_0)\left(\delta^2(1+Q^{-2})^2 + (r_0)^2\sqrt{Q} \right), 
\ee
where $Q$ must satisfy the constraints
$Q\le 1/2$ and 
$
Q> \sqrt{\frac{2\delta}{(r_0-2\delta)}}.$
The best choice of $Q$ is then
$
Q= \left( \frac{8\delta} {r_0}\right)^{4/9}, 
$
which satisfies the constraints under our hypothesis.   
Substituting this value of $Q$ and using calculations
similar to (\ref{A2-2}) we obtain
$$
\vol(A)\le \eps_U(D, r_0, \delta).
$$
Recomputing $vol(B)$ using $r_{D^+}=r_0$ we alter (\ref{B1})
and obtain
$$
\vol(B)\le D \eps_U(D, r_0, \delta).
$$
The same idea gives us (\ref{V2}).   To obtain the estimate on
the volumes of the boundaries of the inner tubular neighborhoods,
observe that 
$
\partial U_D(\Sigma_0)=\Sigma_0 \cup r^{-1}(r_{D^-})
$
and 
$$
\vol(\Sigma_0\subset \E^3)=4\pi r_0^2=\vol(\Sigma_0\subset M^3). 
$$
So, we need only the upper estimate  
$$
\vol(r^{-1}(r_{D^-}))\le \vol(r^{-1}(r_\eps))=4\pi r_\eps^2, 
$$
which is estimated exactly as in the first term of (\ref{V8}).
\end{proof}


\subsection{Nonlinear stability assuming bounded depth}    

Recall the definition of depth in the introduction.
Given a surface $\Sigma$ in a complete and non-compact manifold,
such that $\Sigma=\partial \Omega\setminus \partial M$ we have
\be
\Depth(\Sigma)=\inf\{D: \,\, \Omega \subset T_D(\Sigma) \}, 
\ee
where the infimum is taken over all tubular regions. 

For $(M^m, g) \in \RS$, and $\Sigma_0$ of fixed area $\vol(\Sigma_0)=\alpha_0$ and
$m_H(\Sigma_0)=\delta$, it is possible for the depth to be arbitrarily large. (See  examples in \cite{LS1}.) 
The following statement is a direct consequence of Theorem~\ref{thm-adapted} since
$\Omega_0=Cl(U_D(\Sigma_0))$.   The only difference is that the
boundaries of the regions now match completely.   

\begin{theorem}[An estimate assuming bounded depth]
\label{thm-depth}    
Suppose $(M^3, g) \in \RSonethree$ and $m_{H}(\Sigma_{\alpha_0}) =\delta$ with
$\delta \le r_0/32$, 
where $\Sigma_{0}=\partial \Omega_0$ be the CMC surface of area
$
\alpha_0=4\pi r_0^2
$
and suppose that 
$\Depth(\Sigma)\le D$. 
Then one has 
\be
\aligned
 d_{\mathcal{F}}(\,\Omega_0\subset M^3\,,\, \Omega_0\subset \E^3\,)
& < (1+D)\eps_U(\delta, D, r_0),
\\
d_{D\mathcal{F}}(\,\Omega_0\subset M^3\,,\, \Omega_0\subset \E^3\,)
& < 2\eps_U(\delta, D, r_0),  
\endaligned
\ee
where
$\eps_U(D, r_0,\delta):=
48\pi (2D+\pi r_0)r_0^{16/9}\delta^{2/9}$ 
and, furthermore, 
 \be
\aligned
\left| \vol(\Omega_0\subset M^3)- \vol(\Omega_0\subset \E^3)\right|
& \le \eps_U(D, r_0,\delta),
\\
\vol(\partial \Omega_0\subset M^3))
&=
4\pi r_0^2= \vol(\partial\Omega_0\subset \E^3).
\endaligned
\ee
\end{theorem}
 

\section{Nonlinear stability in the Sobolev norm}
\label{sec:7}

\subsection{Preliminaries}
\label{sec:71}

The $H^1$ Sobolev norm between two diffeomorphic regions in manifolds depends upon the diffeomorphim.  Thus, given a diffeomorphism,
$\Psi: W_1 \to W_2$, the Sobolev norm of interest is $\| \Psi_* g_1 - g_2 \|_{H^1(W_2)}$. 
This norm does not scale when one rescales the manifolds.   In fact, the zero-th order terms scale like the square root of volume times distance squared while the first-order terms seem to scale like square root of volume alone.  We will also use the {\bf D-Sobolev norm}, defined by dividing\footnote{It would also be natural to divide here by the square root of volume.}
 the zero-th order terms by a diameter bound $D$.  

We are interested in controling the Sobolev norm between the inner tubular regions 
$U_D(\Sigma_0)\subset M^3$ and $U_D(\Sigma_0)\subset \E^3$  
for $M^3\in \RSonebarthree$ and $U_D(\Sigma_0)\subset \E^3$ where $\Sigma_0$ is
a CMC surface of area $\alpha_0$.  Bounds on Sobolev norm are not gauge invariant and require a well chosen diffeomorphism.   Here we use the intuition from Theorem~\ref{thm-adapted} to set up a diffeomorphism.

We proceed as follows. First in Section~\ref{sec:72} below, we assume the inner tubular regions are thin in the sense that $D< r_0 =\sqrt{\alpha_0/4\pi}$ since then both $U_D(\Sigma_0)\subset M^3$ and $U_D(\Sigma_0)\subset \E^3$  
are diffeomorphic to annular regions in $\R^3$ and we can set up a simple 
diffeomorphism which preserves the rotational symmetry and preserves the radial lengths.  Next, in Section~\ref{sec:73}, we study the $H^1$ sobolev norm without setting up
diffeomorphisms between $U_D(\Sigma_0)\subset M^3$ and $U_D(\Sigma_0)\subset \E^3$  since these regions need not be diffeomorphic when $D \ge r_0$ depending upon the depth of $\Sigma$.
 

\subsection{Nonlinear Sobolev stability of thin inner tubular regions}
\label{sec:72}

Here we consider thin inner tubular regions $U_D(\Sigma_0)$. Our condition on the mass 

\begin{theorem}[Nonlinear stability of thin regions in the $H^1$ norm] 
\label{thm-sobolev-1}  
Consider spaces $(M^3, g) \in \RSonebarthree$ and $m_{H}(\Sigma_{\alpha_0}) =:\delta$ with
\be 
\delta =m_H(\Sigma_0),
\ee
where $\Sigma_{0}=\partial \Omega_0$ is a CMC surface\footnote{Within the class $\RSonebarthree$, this surface may not be unique since the profile function may be constant on some intervals, but our bounds hold for any choice.}
 of area
$\alpha_0=4\pi r_0^2$.
Let $\sigma(x)=d(x, \Sigma_0)$ so that $s=s(\Sigma_0)-\sigma$.   
If 
$$
D < r_0,
$$
one can define a diffeomorphism
$\Psi: U_D(\Sigma_0)\subset M^3 \to U_D(\Sigma_0)\subset \E^3$ 
such that $\sigma(x)=\sigma(\Psi(x))$ and such that radial geodesics are
isometrically mapped to radial geodesics.   Then at a point $x \in U_D(\Sigma_0)\subset \E^3$ one may evaluate the metrics 
$$
\aligned
\psi_*g &= d\sigma^2 + ( f(s_M-\sigma))^2 g_{\mathbb{S}^2},
\\
g_{\E}&= d\sigma^2 + (s_{\E}-\sigma)^2 g_{\mathbb{S}^2}, 
\endaligned
$$
where $s_M=s(\Sigma_0\subset M)$ and $s_{\E}=s(\Sigma_0\subset \E^3)$. 
Then the Sobolev norm over $U=U_D(\Sigma_0)\subset \E^3$ can be estimated as 
\be
\| \Psi_*g-g_\E \|_{H^1(U)} \le  \sqrt{1+r_0^2}\,\,\eps_{H^1}(D, r_0, \delta)
\ee
and the D-Sobolev norm over $U=U_D(\Sigma_0)\subset \E^3$ can be estimated as 
\be
\aligned
\| \Psi_*g-g_\E \|_{DH^1(U)} & \le  \sqrt{2}\,\,\eps_{H^1}(D, r_0, \delta),
\endaligned
\ee
in which $\eps_{H^1}(D, r_0, \delta) : = 8 \sqrt{\pi}\, r_0^2\, \delta^{1/3}\, D^{1/6}$. 
\end{theorem}

More precisely, we have $\| \Psi_*g-g_\E \|^2_{H^1(U)} = N_0(U) + N_1(U)$ with 
\be
\aligned
N_0(U) :=& \int_{0}^{D} | ( f(s_M-\sigma) )^2 - (s_{\E}-\sigma)) ^2 |^2 \, 4\pi (s_{\E}-\sigma)^2 d\sigma, 
\\
N_1(U) := &  \int_{0}^{D}  \Big| 
(d/d\sigma)\big(f(s_M-\sigma)\big)^2 - (d/d\sigma) \big(s_{\E}-\sigma \big)^2 
\Big|^2 
\, 4\pi (s_{\E}-\sigma)^2 d\sigma. \\
\endaligned
\ee
Observe that $N_0=0$ and $N_1=0$ precisely when 
\be
f(s_M-\sigma) = (s_{\E}-\sigma),   \qquad  \sigma\in [0,D], 
\ee
which occurs if and only if the map 
$\Psi: U_D(\Sigma_0)\subset M \to U_D(\Sigma_0)\subset \E$
is an isometry.

We will follow the following heuristics. We wish to show that $N_0+N_1$ is small when $m_H(\Sigma_0)$
is small.  Motivated by \cite{LS1} where a radius $r_\eps$ near $0$ was chosen to cut out the well, we will select a suitable 
$\sigma_\eps$ close to $D$ to cut out the well.  One cannot make the metric small in the well so the
integrals for $\sigma\in [\sigma_\eps, D]$ will be bounded by the volume of the region.   For $\sigma\in [0, \sigma_\eps]$, the smallness of the Hawking mass will control the metric.

\begin{proof} 1. In the beginning of this proof we will not use the fact that $D\le r_0$, 
so that we may also use these estimates in the following sections.  
Recall that 
$
f(s_M)=r_0=s_\E
$
but that $s_M$ might be much much larger that $s_\E$ if the
depth of $\Sigma_0$ is very large. Furthermore $f$ is monotone increasing and
\be
f'(s) = \sqrt{1 - \frac{2m_H(s)}{f(s)} } \le 1.
\ee
Thus, we find  
\be
 \label{f-S-M}
\aligned
f(s_M-\sigma)
&= f(s_M)-\int_{s_M-\sigma}^{s_M} f'(s)\,ds\\
&\ge s_{\E}-\int_{s_M-\sigma}^{s_M} 1 \,ds
= s_{\E} -(s_M-(s_M-\sigma))=s_\E-\sigma.
\endaligned
\ee

In order to derive the Sobolev estimates, we will break our integration at some $\sigma_\eps \in [0, r_0]$: 
\be
f(s_M -\sigma)\ge f(s_M-\sigma_\eps), \qquad  \sigma \in [0, \sigma_\eps].
\ee
Thus, we have 
$$
f'(s_M -\sigma) \ge \sqrt{1 - \frac{2m_H(s_M-\sigma)}{f(s_m-\sigma_\eps)}},  
\qquad  \sigma \in [0, \sigma_\eps]
$$
and, since the Hawking mass is non-decreasing as well,
$$
f'(s_M -\sigma) \ge \sqrt{1 - \frac{2m_H(s_M)}{f(s_m-\sigma_\eps)}}, 
\qquad \sigma \in [0, \sigma_\eps].
$$
Thus, we find 
\be
1 \ge f'(s_M-\sigma) \ge \sqrt{1 - \frac{2\delta}{f(s_m-\sigma_\eps)}}, 
\qquad \sigma \in [0, \sigma_\eps] 
\ee
and 
$$
\aligned
\big| d/d\sigma \left(f(s_M-\sigma)-(s_{\E}-\sigma)\right) \big|
&=|-f'(s_M-\sigma)+1|=-f'(s_M-\sigma) +1
\\
&\le 1-\sqrt{1 - \frac{2\delta}{f(s_M-\sigma_\eps)}}
 \le E(\delta, \sigma_\eps), \qquad \qquad 
 \sigma\in [0,\sigma_\eps], 
\endaligned
$$
where we have introduced  the scale invariant function 
\be\label{Edefn}
E(\delta, \sigma_\eps)
:=1-\sqrt{1 - \frac{2\delta}{(s_M-\sigma_\eps)} } =
1-\sqrt{1 - \frac{2\delta}{(r_0-\sigma_\eps)}}.
\ee
Here, we have applied (\ref{f-S-M}) and $s_M=r_0$ in order to obtain the final line in this estimate.

Also, we have 
$$
\aligned
f(s_M-\sigma)-(s_{\E}-\sigma)
&=
f(s_M)-(s_{\E}) + \int_{a=0}^{\sigma} 
d/da \,\left( f(s_M-a)-(s_{\E}-a)\right) \,da 
\\
&= 0+ \int_{a=0}^{\sigma} 
-f'(s_M-a)+1 \,da
\endaligned
$$
and so 
\be
\label{DE}
\aligned
|f(s_M-\sigma)-(s_{\E}-\sigma)|
&\le
\int_{a=0}^{\sigma} |-f'(s_M-a)+1| \,da 
\\
&\le\sigma_\eps E(\delta, \sigma_\eps) \le r_0E(\delta,\sigma_\eps), 
\qquad \qquad \sigma \in [0, \sigma_\eps]. 
\endaligned
\ee
It follows that
$$
\aligned
| ( f(s_M-\sigma) )^2 - (s_{\E}-\sigma)) ^2 |^2 
&=
| ( f(s_M-\sigma) ) - (s_{\E}-\sigma)) |^2 | ( f(s_M-\sigma) ) + (s_{\E}-\sigma)) |^2 \\
 &\le  | r_0E(\delta,\sigma_\eps) |^2 | 2r_0 |^2.
 \endaligned
$$

Then, we can also bound
$$
\aligned
 \left| 
\frac{d}{d\sigma}\big(f(s_M-\sigma)\big)^2 - \frac{d}{d\sigma} \big(s_{\E}-\sigma \big)^2 
\right| 
&=
 \Big|  2 f(s_M-\sigma) f'(s_M-\sigma) - 2(s_{\E}-\sigma) \Big| \\
& \le
 \Big|  2 f(s_M-\sigma) f'(s_M-\sigma) -2 f(s_M-\sigma)\Big| +
 \Big|2f(s_M-\sigma)- 2(s_{\E}-\sigma) \Big| \\
& \le 2 r_0 |  f'(s_M-\sigma) -1| + 2  |f(s_M-\sigma)-(s_\E-\sigma)|, 
\endaligned
$$
hence
\be
\label{primeE}
\aligned
 \left| 
\frac{d}{d\sigma}\big(f(s_M-\sigma)\big)^2 - \frac{d}{d\sigma} \big(s_{\E}-\sigma \big)^2 
\right| 
& \le 2 r_0 E(\delta, \sigma_\eps) + 2 r_0E(\delta, \sigma_\eps), 
\qquad\sigma\in [0, \sigma_\eps].
\endaligned
\ee

\vskip.15cm

2. We may now apply these estimates to approximate $N_0$ and $N_1$.
First observe that
$$
\aligned
N_0(U)&\le N_0(U_\eps)+N_0(U\setminus U_\eps),\\
N_1(U)&\le N_1(U_\eps)+N_1(U\setminus U_\eps),
\endaligned
$$
where
\be
\aligned
N_0(U_\eps) :=& \int_{0}^{\sigma_\eps} | ( f(s_M-\sigma) )^2 - (s_{\E}-\sigma)) ^2 |^2 \, 4\pi (s_{\E}-\sigma)^2 d\sigma, 
\\
N_0(U\setminus U_\eps) :=& \int_{\sigma_\eps}^{D} | ( f(s_M-\sigma) )^2 - (s_{\E}-\sigma)) ^2 |^2 \, 4\pi (s_{\E}-\sigma)^2 d\sigma. 
\\
N_1(U_\eps) :=& \int_{0}^{\sigma_\eps} 
 \Big| 
(d/d\sigma)\big(f(s_M-\sigma)\big)^2 - (d/d\sigma) \big(s_{\E}-\sigma \big)^2 
\Big|^2 
\, 4\pi (s_{\E}-\sigma)^2 d\sigma
\\
N_1(U\setminus U_\eps) :=& \int_{\sigma_\eps}^{D}
 \Big| 
(d/d\sigma)\big(f(s_M-\sigma)\big)^2 - (d/d\sigma) \big(s_{\E}-\sigma \big)^2 
\Big|^2 
\, 4\pi (s_{\E}-\sigma)^2 d\sigma.
\endaligned
\ee
Our estimates on $N_0(U_\eps)$ and $N_1(U_\eps)$ hold for any choice of $\sigma_\eps\in (0, r_0)$ which gives us (\ref{Edefn})-(\ref{primeE}) 
and will be used in the following as well.  
 So we
find these estimates first:
$$
\aligned
N_0(U_\eps)&= \int_{0}^{\sigma_\eps} | ( f(s_M-\sigma) )^2 - (s_{\E}-\sigma)) ^2 |^2 \, 4\pi (s_{\E}-\sigma)^2 d\sigma\\
 &\le 
 \int_{0}^{\sigma_\eps} | r_0E(\delta,\sigma_\eps) |^2 | 2r_0 |^2 
 \, 4\pi (r_0-\sigma)^2 d\sigma \\
&\le  |r_0E|^2|2r_0|^2
 \int_{r_0-\sigma_\eps}^{r_0}
 \, 4\pi r^2 dr \\
&\le  |r_0E(\delta,\sigma_\eps)|^2|2r_0|^2 \vol(Ann_{0}(r_0-\sigma_\eps, r_0)\subset \E^3),
\endaligned
$$
thus 
\be
N_0(U_\eps) \le  |r_0E(\delta,\sigma_\eps)|^2|2r_0|^2 \vol(B_{0}(r_0)\subset \E^3).
\ee

 The following estimate on $N_1(U_\eps)$ also
 holds for any choice of $\sigma_\eps$ which gives us (\ref{Edefn})-(\ref{primeE}): 
\be
\aligned
N_1(U_\eps) &=  
\int_{0}^{\sigma_\eps} 
 \Big| 
(d/d\sigma)\big(f(s_M-\sigma)\big)^2 - (d/d\sigma) \big(s_{\E}-\sigma \big)^2 
\Big|^2 
\, 4\pi (s_{\E}-\sigma)^2 d\sigma\\
&=
\int_{0}^{\sigma_\eps} 
 \Big|  2 r_0 E(\delta, \sigma_\eps) + 2 r_0E(\delta, \sigma_\eps) \Big|^2 
\, 4\pi (r_0-\sigma)^2 d\sigma\\
&\le 4^2r_0^2(E(\delta, \sigma_\eps))^2  \vol(Ann_{0}(r_0-\sigma_\eps, r_0)\subset \E^3)
\endaligned
\ee
and, therefore, 
\be
\label{reuse}
N_1(U_\eps) \le 16r_0^2(E(\delta, \sigma_\eps))^2  \vol(B_{0}(r_0)\subset \E^3).
\ee

\vskip.15cm 

3. The rest of the proof of this theorem which estimates the inner regions relies heavily on $D< r_0$ and will not be used in the proofs of subsequent
theorems.

Our estimate on $N_0(U\setminus U_\eps)$ cannot apply the strong controls on the metric provided in (\ref{DE}) but instead will rely on the
small volume of the regions and use the fact that $D<r_0$:
$$
\aligned
N_0(U\setminus U_\eps)&= \int_{\sigma_\eps}^D | ( f(s_M-\sigma) )^2 - (s_{\E}-\sigma)) ^2 |^2 \, 4\pi (s_{\E}-\sigma)^2 d\sigma\\
&\le \int_{\sigma_\eps}^D | ( f(s_M-\sigma) )^2 + (s_{\E}-\sigma)) ^2 |^2 \, 4\pi (s_{\E}-\sigma)^2 d\sigma\\
&\le \int_{\sigma_\eps}^D |  r_0^2 + r_0^2 |^2 \, 4\pi (s_{\E}-\sigma)^2 d\sigma
\le |  r_0^2 + r_0^2 |^2 \, \int_{r_0-D}^{r_0-\sigma_\eps} 4\pi r^2 dr
\\
 &\le | 2 r_0^2 |^2 \, \vol(Ann_0(r_0-D,r_0-\sigma_\eps),
\endaligned
$$
thus
\be
N_0(U\setminus U_\eps) 
\le | 2 r_0^2 |^2 \, 4\pi (r_0-\sigma_\eps)^2(D-\sigma_\eps)
\le | 2 r_0^2 |^2 \, 4\pi (r_0-\sigma_\eps)^2 D.
\ee

Our estimate on $N_0(U\setminus U_\eps)$ cannot apply the strong controls on the metric provided in (\ref{DE}) but instead will rely on the
small volume of the regions and $f'\le 1$ and use the fact that $D<r_0$:
$$
\aligned
N_1(U\setminus U_\eps) &= \int_{\sigma_\eps}^{D}
 \Big| 
(d/d\sigma)\big(f(s_M-\sigma)\big)^2 - (d/d\sigma) \big(s_{\E}-\sigma \big)^2 
\Big|^2 
\, 4\pi (s_{\E}-\sigma)^2 d\sigma\\
&=
\int_{\sigma_\eps}^{D}
\Big|  2 f(s_M-\sigma) f'(s_M-\sigma) - 2(s_{\E}-\sigma) \Big|^2 
\, 4\pi (s_{\E}-\sigma)^2 d\sigma\\
&\le
\int_{\sigma_\eps}^{D}
\Big|  2 f(s_M-\sigma) f'(s_M-\sigma) + 2(s_{\E}-\sigma) \Big|^2 
\, 4\pi (s_{\E}-\sigma)^2 d\sigma, 
\endaligned
$$
thus
\be
N_1(U\setminus U_\eps) 
\le
\int_{\sigma_\eps}^{D}
\Big|  2 r_0 (1 )+ 2 r_0 \Big|^2 
\, 4\pi (s_{\E}-\sigma)^2 d\sigma
\le | 4 r_0 |^2 \, 4\pi (r_0-\sigma_\eps)^2 D. 
\ee
Combining all of these estimates we have
$$
\aligned
N_0(U)&\le N_0(U_\eps)+N_0(U\setminus U_\eps)\\
&\le |r_0E(\delta,\sigma_\eps)|^2 4r_0^2(4/3)\pi r_0^3
+ 4r_0^2r_0^2 \, 4\pi (r_0-\sigma_\eps)^2 D \\
&\le16\pi r_0^4 \,  
[(E(\delta, \sigma_\eps))^2r_0^3 + (r_0-\sigma_\eps)^2D]
\endaligned
$$
and 
$$
\aligned
N_1(U)&\le  N_1(U_\eps)+N_1(U\setminus U_\eps)\\
&\le 16r_0^2(E(\delta, \sigma_\eps))^2  (4/3)\pi r_0^3
+  16r_0^2 \, 4\pi (r_0-\sigma_\eps)^2 D\\
&\le 32\pi r_0^2 \,  
[(E(\delta, \sigma_\eps))^2r_0^3 + (r_0-\sigma_\eps)^2D].
\endaligned
$$
So whether we wish to estimate the Sobolev norm $\sqrt{N_0(U)+N_1(U)}$
or the $D$ Sobolev norm $\sqrt{N_0(U)/(r_0+D)^2 + N_1(U)}$, we must
choose a good estimate for
$$
\aligned
F(\sigma_\eps)&
:=
(E(\delta, \sigma_\eps))^2r_0^3 + (r_0-\sigma_\eps)^2D\\
&=\left(1-\sqrt{1 - \frac{2\delta}{(r_0-\sigma_\eps)}\, } \right)^2r_0^3 + (r_0-\sigma_\eps)^2D\\
&\le\left(1-\sqrt{1 - \frac{2\delta}{(r_0-\sigma_\eps)}\, } \right)
\left(1+\sqrt{1 - \frac{2\delta}{(r_0-\sigma_\eps)}\, } \right)
r_0^3 + (r_0-\sigma_\eps)^2D\\
&=\left(1- \left(1 - \frac{2\delta}{(r_0-\sigma_\eps)}\, \right) \right)
r_0^3 + (r_0-\sigma_\eps)^2D
=\frac{2\delta}{(r_0-\sigma_\eps)}\, 
r_0^3 + (r_0-\sigma_\eps)^2D. 
\endaligned
$$
Observe that 
$$
\aligned
&0=F'(\sigma)=\frac{2\delta}{(r_0-\sigma)^2}\, r_0^3 -2(r_0-\sigma)D,
\qquad 
&
2(r_0-\sigma)D=\frac{2\delta}{(r_0-\sigma)^2}\,r_0^3,
 \\
&2(r_0-\sigma)^3D=2\delta\,r_0^3, 
\qquad 
&
(r_0-\sigma)=\left(\delta\,r_0^3/D\right)^{1/3}.  
\endaligned
$$

Now, since $\sigma_\eps \in [0,D]\subset [0, r_0]$, we distinguish between two cases:
$$
\aligned
\textrm{Case I: } & r_0 -\left(\delta\,r_0^3/D\right)^{1/3} \le D.
\\
\textrm{Case II: } & r_0 -\left(\delta\,r_0^3/D\right)^{1/3} > D. 
\endaligned
$$
In Case I, we take $\sigma_\eps= r_0 -\left(\delta\,r_0^3/D\right)^{1/3}$
and obtain 
\be
F(\sigma_\eps)= \frac{2\delta}{\left(\delta\,r_0^3/D\right)^{1/3} }\, 
r_0^3 + (\delta\,r_0^3/D )^{2/3}D= 2\delta^{2/3}\,r_0^2 D^{1/3}, 
\ee
thus
\bel{eq:567-I}
\aligned
N_0(U) &\le  16\pi r_0^4 [2\delta^{2/3}\,r_0^2 D^{1/3}],
\\
N_1(U)&\le 32\pi r_0^2 [2\delta^{2/3}\,r_0^2 D^{1/3}].
\endaligned
\ee
On the other hand, in Case II, we take $\sigma_\eps=D$, so that
$$
\aligned
N_0(U)\le N_0(U_\eps)+0
&\le |r_0E(\delta,D)|^2 4r_0^2(4/3)\pi r_0^3\\
&\le16\pi r_0^4 \,  
[(E(\delta, D))^2r_0^3]\\
&\le16\pi r_0^4 \,  
[2\delta/(r_0-D)r_0^3]\\
&\le16\pi r_0^4 \,  
[2\delta/(\delta r_0^3/D)^{1/3}r_0^3]  \le 
 16\pi r_0^4 \,  
[2\delta^{2/3}D^{1/3}r_0^2D^{1/3}], 
\endaligned
$$
where the condition in Case II was used in the penultimate inequality 
and 
$$
\aligned
N_1(U)\le  N_1(U_\eps)+0
&\le 32\pi r_0^2 \,  [(E(\delta, \sigma_\eps))^2r_0^3]\\
&\le  32\pi r_0^2
[2\delta^{2/3}D^{1/3}r_0^2D^{1/3}], 
\endaligned
$$
so that we now find 
\bel{eq:567-II}
\aligned
N_0(U) &\le 
 16\pi r_0^4 \,  [2\delta^{2/3}D^{1/3}r_0^2D^{1/3}], 
\\
N_1(U)&\le  32\pi r_0^2[2\delta^{2/3}D^{1/3}r_0^2D^{1/3}].
\endaligned
\ee
This completes the proof of Theorem~\ref{thm-sobolev-1}. 
\end{proof}


\subsection{Nonlinear Sobolev stability without diffeomorphisms}
\label{sec:73}

Here we would like to compare regions $U_D(\Sigma)$
which may not be diffeomorphic.  To do so we
define the backward profile function, $h$, emanating from $\Sigma$ as follows and estimate the Sobolev bounds on $h^2$ rather than
setting up a diffeomorphism.   

\begin{definition} \label{defn-back}
Fix $r_0 >0$. 
Given a manifold $M$  in $\RSzerobarthree$ with profile function $f$ and given any CMC hypersurface\footnote{Again, this surface may not be unique.}
 $\Sigma_0$ with area 
$\alpha_0=4\pi r_0^2$ one considers the parameter value $s_M\ge 0$ such that 
$f(s_M)=r_0$ and define the {\bf backward profile function} (determined from the hypersurface $\Sigma_0$) to be 
\be
h:[0,\infty) \to [0,\infty), 
\qquad \qquad 
h(\sigma)
=
\begin{cases}
f(s_M-\sigma), \qquad &  \sigma\le s_M,
\\
f(0),             & \sigma> s_M,
\end{cases} 
\ee
which is monotone non-increasing and may be discontinuous. 
\end{definition}

When the additional regularity $(M^3, g) \in \RSonebarthree$ is assumed, then the {\bf backward profile function} is actually (Lipschitz) continuous. Furthermore, the regularity $f \in H^1$ implies the same regularity $h \in H^1$ for the backward profile.

In addition, observe that, in Euclidean space, we have $f(s)=s$ and
$h_\E(\sigma)=\max\{(r_0-\sigma),0\}$, 
which is positive only on $[0, r_0)$. 
Example~\ref{thin-well} below will give an explanation as to why it is essential
to consider here these backward profile functions rather than the original functions.

\begin{theorem}[Nonlinear stability in the $H^1$ norm]
\label{thm-sobolev-no-diffeo}  
Consider a space $(M^m, g) \in \RS^{\weak, 1}_m$ with $m_{H}(\Sigma_{\alpha_0}) =:\delta$ and
$
\delta =m_H(\Sigma_0),  
$
where $\Sigma_{0}$ denotes any CMC surface of area
$\alpha_0=4\pi r_0^2$. 
Then for any $D>0$, the following estimate holds: 
\be
\aligned
 \|h^2(\sigma) - (r_0-\sigma)^2\|_{H^1[0,D]} &\le \sqrt{1+r_0^2} \, 
\eps_{H^1}(D, r_0, \delta),
\\
 \eps_{H^1}(D, r_0, \delta) &= 16 \sqrt{\pi}\, r_0^2\, \delta^{1/3}\, D^{1/6}.
\endaligned
\ee
\end{theorem}

This estimate when $D<r_0$ was already proven in Theorem~\ref{thm-sobolev-1}, and this new estimate is relevant to cover ``large'' values of $D$. 

\begin{proof}
We must estimate:
$\|h^2(\sigma) - (r_0-\sigma)^2 \|_{H^1[0,D]} =N_0(U) + N_1(U)$,
where, as in Theorem~\ref{thm-sobolev-1} and with $s_{\E}=r_0$, 
\be
\aligned
N_0(U) = & \int_{0}^{D} | ( h(\sigma) )^2 - (s_{\E}-\sigma)) ^2 |^2 \, 4\pi (s_{\E}-\sigma)^2 d\sigma, 
\\
N_1(U) = &  \int_{0}^{D}  \Big| 
(d/d\sigma)\big(h(\sigma)\big)^2 - (d/d\sigma) \big(s_{\E}-\sigma \big)^2 
\Big|^2 
\, 4\pi (s_{\E}-\sigma)^2 d\sigma. 
\endaligned
\ee
  As before we introduce an arbitrary $\sigma_\eps\in (0, r_0)$ and break the integrals
at $\sigma=\sigma_\eps$:
\be
\aligned
N_0(U)&\le N_0(U_\eps)+N_0(U\setminus U_\eps),
\\
N_1(U)&\le N_1(U_\eps)+N_1(U\setminus U_\eps), 
\endaligned
\ee
where
\be
\aligned
N_0(U_\eps) :=& \int_{0}^{\sigma_\eps} | ( h(\sigma )^2 - (s_{\E}-\sigma)) ^2 |^2 \, 4\pi (s_{\E}-\sigma)^2 d\sigma, 
\\
N_0(U\setminus U_\eps) :=& \int_{\sigma_\eps}^{D} | ( h(\sigma) )^2 - (s_{\E}-\sigma)) ^2 |^2 \, 4\pi (s_{\E}-\sigma)^2 d\sigma,
\\
N_1(U_\eps) :=& \int_{0}^{\sigma_\eps} 
 \Big| 
(d/d\sigma)\big(h(\sigma)\big)^2 - (d/d\sigma) \big(s_{\E}-\sigma \big)^2 
\Big|^2 
\, 4\pi (s_{\E}-\sigma)^2 d\sigma,
\\
N_1(U\setminus U_\eps) :=& \int_{\sigma_\eps}^{D}
 \Big| 
(d/d\sigma)\big(h(\sigma)\big)^2 - (d/d\sigma) \big(s_{\E}-\sigma \big)^2 
\Big|^2 
\, 4\pi (s_{\E}-\sigma)^2 d\sigma.
\endaligned
\ee
Choosing $\sigma_\eps< r_0$ in a way 
 which gives us (\ref{Edefn})-(\ref{primeE}), allows us to estimate two of the integrals as in (\ref{reuse}):
$$
\aligned
N_0(U_\eps)&\le |r_0E(\delta,\sigma_\eps)|^2|2r_0|^2 \vol(B_{0}(r_0)\subset \E^3),
\\
N_1(U_\eps) 
&\le 16r_0^2(E(\delta, \sigma_\eps))^2  \vol(B_{0}(r_0)\subset \E^3).
\endaligned
$$

Next, we estimate  
$$
\aligned
N_0(U\setminus U_\eps) 
&= \int_{\sigma_\eps}^{D} 
| ( h(\sigma) )^2 - (s_{\E}-\sigma)) ^2 |^2 \, 4\pi (s_{\E}-\sigma)^2 d\sigma
\\
&= \int_{\sigma_\eps}^{D} 
|( h(\sigma) )^2 + (s_{\E}-\sigma)) ^2 |^2 \, 4\pi (s_{\E}-\sigma)^2 d\sigma
\\
&= \int_{\sigma_\eps}^{D} 
\big( 2 \, |2r_0^2|^2 + 2 f(0)^2 \big) \, 4\pi (s_{\E}-\sigma)^2 d\sigma
\le \big( 2 \, |2r_0^2|^2 + 2 r_0^2 \big)  \, 4\pi (r_0-\sigma_\eps)^2 D.
\endaligned
$$
Finally we use $|h'(\sigma)|\le 1$ to estimate
$$
\aligned
N_1(U\setminus U_\eps) 
&= \int_{\sigma_\eps}^{D}
 \Big| 
(d/d\sigma)\big(h(\sigma)\big)^2 - (d/d\sigma) \big(s_{\E}-\sigma \big)^2 
\Big|^2 
\, 4\pi (s_{\E}-\sigma)^2 d\sigma\\
&= \int_{\sigma_\eps}^{D}
 \Big| 
(2(h(\sigma)h'(\sigma) + 2\big(s_{\E}-\sigma) 
\Big|^2 
\, 4\pi (s_{\E}-\sigma)^2 d\sigma\\
&= \int_{\sigma_\eps}^{D}
 \Big| 2 r_0 (1) + 2r_0
\Big|^2 
\, 4\pi (s_{\E}-\sigma)^2 d\sigma
\le 
 | 4 r_0 |^2 \, 4\pi (r_0-\sigma_\eps)^2 D. 
\endaligned
$$
These are almost the same estimates as in Theorem~\ref{thm-sobolev-1} and it is not difficult to check that 
\eqref{eq:567-I} and {eq:567-II} should now be replaced by 
\bel{eq:567-I-new}
\aligned
N_0(U) &\le  16\pi (r_0^4 + 2 r_0^2) [2\delta^{2/3}\,r_0^2 D^{1/3}],
\\
N_1(U)&\le 32\pi r_0^2 [2\delta^{2/3}\,r_0^2 D^{1/3}], 
\endaligned
\ee
and 
\bel{eq:567-II-new}
\aligned
N_0(U) &\le 
 16\pi (r_0^4 + 2r_0^2) \,  [2\delta^{2/3}D^{1/3}r_0^2D^{1/3}], 
\\
N_1(U)&\le  32\pi r_0^2[2\delta^{2/3}D^{1/3}r_0^2D^{1/3}], 
\endaligned
\ee
in Cases I and II, 
respectively. 
Again we reach the desired conclusion.
\end{proof}


\section{Compactness theorems}
\label{sec:8}

\subsection{Main compactness result}
\label{sec:81}

We now address the issue of the (pre-)compactness of sequences of rotationally symmetric spaces. In contrast with earlier results stated in  $\RSonebar$ which remain also valid in $\RSone$, it is now essential to work within $\RSonebar$ and therefore allow for profile functions that are only {\sl non-decreasing}  and, in other words, we must allow interior closed minimal surfaces. 

Specifically, in this section we prove the following compactness theorem. 
 
\begin{theorem}[Compactness framework in the intrinsic flat distance]  
\label{compactness}
Fix some constants $A_0, D_0, M_0>0$. Consider a sequence of spaces 
$U_j \subset M_j  \in \RSonebar$, where $\partial U_j\setminus \partial M_j$ is a rotationally symmetric surface $\Sigma_j\in M_j$ satisfying 
\be
\Area(\Sigma_j) =A_0,
\ee
\be
\Depth(\Sigma_j) \le D_0,
\ee
\be
\label{Hawking-bound}
m_H(\Sigma_{j}) \le M_0.
\ee
Then a subsequence (also denoted $M_j$) converges in the intrinsic flat sense to
a region $U_\infty\subset M_\infty  \in \RSonebar$.  By taking
$\Sigma_\infty =\partial U_\infty \in M_\infty$, one has the following 
\be\label{c-area}
\Area(\Sigma_\infty) =A_0, 
\ee
\be\label{c-depth}
\Depth(\Sigma_\infty) \le \liminf_{j \to +\infty} \Depth(\Sigma_j) \le D_0,
\ee
\be\label{c-volume}
\vol(U_\infty)=\lim_{j\to\infty}\vol(U_j)\le A_0D_0,
\ee
and
\be\label{c-Hawking}
m_H(\Sigma_{\infty}) =\lim_{j \to +\infty} m_H(\Sigma_j) \le M_0. 
\ee
\end{theorem}

Before we can give a proof of this result, we are going to consider the metrics based at the surface 
$\Sigma_0$ viewed using the backward profile functions denoted by $h_j$, and we will prove that this sequence 
$h_j$ is compact in the strong $H^1$ sense and that the nonnegative scalar curvature condition is preserved; cf.~Proposition~\ref{compSobolev}, below.   This theorem introduces a reversed backwards limit profile function, which we will use to define the
limit $U_\infty$ introduced in Theorem~\ref{compactness} above.

 Next, in Section~\ref{sec:83} below, we will exhibit an intrinsic flat limit by applying Wenger's flat compactness theorem.  Finally, by combining these observations, we will construct an isometry between the Sobolev and flat limits, and arrive at the desired compactness theorem in
the intrinsic flat distance, with the property that the nonnegative scalar curvature condition is retained in the limit.

In Example~\ref{no-scalar} below, we will show that while the notion of nonnegative scalar curvature in the sense of distributions persists under intrinsic flat convergence, scalar curvature is not converging.


\subsection{Compactness in the Sobolev norm}
\label{sec:82}

The following theorem is of interest in its own sake and will also be used in order to construct the limit space arising in
Theorem~\ref{compactness}.

\begin{theorem} [Compactness framework in the Sobolev norm]  
\label{compSobolev}  
Fix some constant $M_0$ and consider any sequence of spaces $(M_j, g_j) \in \RSonebar$ with profile functions $f_j$ 
satisfying the following uniform ADM mass bound: 
\be
m_{ADM}(M_j,g_j) \leq M_0. 
\ee
Then, the following properties hold:   
\bei
\item{\bf Backward profile functions.}  Fix some area $A_0=4\pi r_0^2 >0$ and consider the backward profile functions $h_j$ associated with the radius $r_0$. Then, the function $h_j$ and its derivative subconverge pointwise to a limit $h_\infty$ which is non-increasing and Lipschitz continuous: 
\be
\aligned
h_j(\sigma) \to h_\infty(\sigma) \quad \text{ at every } \sigma,  
\\
h_j'(\sigma) \to h'_\infty(\sigma) \quad \text{ at every } \sigma.   
\endaligned
\ee
In particular, these convergence properties imply the strong convergence $h_j \to h_\infty$ in the $H^1$ norm.

\item{\bf Reversed backwards limit profile function.} Assume, in addition, a uniform upper bound on the depth of 
a level set $\Sigma_0\subset M_j$, whose area equals $\Area(\Sigma_0)=4\pi r_0^2$,
\be
\Depth(\Sigma_0\subset M_j)\le D_0. 
\ee
Then, $h_j(\sigma)=0$ for $\sigma>D_0$ so that the
same property holds for $h_\infty$.  This allows us to define a reversed
backwards profile limit
\be
f(s) := h_\infty(s_\infty-s) \quad \textrm{ with } f(0)=r_{min \infty},
\ee
in which 
\be
s_\infty :=\sup\{\sigma: \, h_\infty(\sigma)>0\} \le D_0, 
\qquad 
r_{min \infty}=\lim_{\sigma' \to s_\infty} h_\infty(\sigma'). 
\ee 
This function precisely satisfies the conditions of a profile function for a space lying in $\RSonebar$ restricted to $[0, s_\infty]$
and the Hawking mass functions $m_{H j}$ of the spaces $M_j$ also converge pointwise. 
\eei 
\end{theorem}

At this stage, it is important to emphasize the following: 
\bei 
\item In Example~\ref{thin-well} below, we illustrate why the limit of the original functions
$f_j$ is not as geometrically natural as the reversed backwards profile
limit $f$.    

\item Namely, it may happen that the functions $f_j$ converge to $0$ while the functions $h_j$ converge to the Euclidean space's backward profile function, so that the reversed backwards profile limit is $f(s)=s$.

\item This observation is consistent with our conclusion above which {\sl does not} claim that $h_\infty$ is the backward profile function associated with $f_\infty$. 
\eei 

\begin{proof} In view of the regularity property \eqref{eq:215} and since the Hawking mass is uniformly bounded, we have 
$h_j' \in BV_\loc(0, +\infty)$ together with a uniform bound on the total variation of the functions $h_j' \in [0, -1]$. Therefore, by Helly's theorem \cite{EG}, a subsequence of $h_j'$ converges at every $s$ to some limit denoted by $h_\infty'$. This convergence property consequently holds in any $L^p$  norm for $p \in [1,+\infty]$.  
In addition, by construction, the functions $h_j \geq 0$ are uniformly bounded (since $h_j(r_0)=0$ and $h_j$ is non-increasing)
 and, therefore, 
converge uniformly, as follows:  
\bel{eq:unif}
\sup_\sigma |h_j(\sigma) - h_\infty(\sigma)| \leq \int_0^{+\infty} |h_j'(\sigma) - h_\infty'(\sigma)| \, d\sigma \to 0. 
\ee
In particular, this pointwise convergence of $h_j$ and $h_j'$ implies the convergence $h_j \to h_\infty$ in the $H^1$ norm. 
 
Furthermore, let us consider the Hawking mass functions $m_{H j}$. By assumption, these functions are nonnegative and non-increasing and are uniformly bounded by the ADM mass. Therefore, they converge to a limit $m_{H \infty}$ which is also non-increasing and satisfies 
$$
0 \leq m_{H \infty} \leq M_0. 
$$
Furthermore, importantly, in view of \eqref{eq:220}, we have 
$$
2 \, m_{H j}(\sigma) = (h_j(\sigma))^{m-2} \big( 1 - (h_j'(\sigma))^2 \big),
$$
in which the right-hand side converges pointwise, 
so that this limit can also be regarded as the Hawking mass associated with the function $h_\infty$, that is, 
\bel{eq:limmass}
2 \, m_{H \infty}(\sigma) = (h_\infty(\sigma))^{m-2} \big( 1 - (h_\infty'(\sigma))^2 \big). 
\ee

Now, the function $f$ defined as the statement fo the theorem clearly satisfies the regularity conditions of a profile function for a space
lying in $\RSonebar$.  Also, since $m_{H \infty}$ is non-increasing, this space has nonnegative scalar curvature. 
\end{proof}

\begin{remark} \rm
\label{lengths}
In the following section, we will use the following consequence of Theorem~\ref{compSobolev}:  for any curve $(\theta(t), \sigma(t))$, the length 
$$
\int_0^1 \sqrt{|\sigma'(t)|^2 + h_j^2(\sigma(t))|\theta'(t)|^2} \, dt
$$
converges to 
$$
\int_0^1 \sqrt{|\sigma'(t)|^2 + h_\infty^2(\sigma(t))|\theta'(t)|^2} \, dt. 
$$
This is immediate from the uniform convergence property $h_j \to h_\infty$ in \eqref{eq:unif}. 
\end{remark}


\subsection{Sobolev to intrinsic flat compactness}
\label{sec:83}

In light of Theorem~\ref{compSobolev}, in order to complete the proof of Theorem~\ref{compactness} we need only check the following result.  

\begin{proposition} \label{compIF}
Given a sequence as in Theorem~\ref{compactness}
we obtain a sequence of profile functions as in Proposition~\ref{compSobolev} whose backwards profile functions converge allowing us to define
a limit space $U_\infty\subset M_\infty  \in \RSonebar$ 
using the reversed backwards limit profile function $f$.  Then, one has 
\be\label{eq-compIF}
d_{\mathcal{F}}(U_j, U_\infty) \to 0
\ee
and, by taking
$\Sigma_\infty =\partial U_\infty \in M_\infty$, the conditions 
(\ref{c-area})-(\ref{c-Hawking}) hold.
\end{proposition}

To prove this statement, we first observe that the following theorem (established first in Lakzian and Sormani \cite{LakzianSormani} for sufficiently regular spaces) holds even when $g_i\in \RSzerobar$, thanks to our work in Proposition~\ref{prop-int-curr-space} above. 

\begin{theorem} \label{thm-subdiffeo}  (See \cite{LakzianSormani}.) 
Suppose $M_1=(M,g_1)$ and $M_2=(M,g_2)$ are oriented
precompact Riemannian manifolds
with diffeomorphic subregions $W_i \subset M_i$ and
diffeomorphisms $\psi_i: W \to W_i$ such that
\be  
\aligned
\psi_1^*g_1(V,V)
& < (1+\eps)^2 \psi_2^*g_2(V,V) \qquad \text{ for all } \, V \in TW, 
\\  
\psi_2^*g_2(V,V) 
&<
(1+\eps)^2 \psi_1^*g_1(V,V) \qquad  \text{ for all } \, V \in TW.
\endaligned
\ee
Taking the extrinsic diameters, i.e. 
$$ 
D_{W_i}= \sup \big\{\diam_{M_i}(W): \, W\textrm{ is a connected component of } W_i \big\} \le \diam(M_i),
$$
one can introduce the hemispherical width 
\be \label{thm-subdiffeo-3}
a>\frac{\arccos(1+\eps)^{-1} }{\pi}\max\{D_{W_1}, D_{W_2}\}.
\ee
Taking the difference in distances with respect to the outside manifolds,
\be \label{lambda}
\lambda=\sup_{x,y \in W}
|d_{M_1}(\psi_1(x),\psi_1(y))-d_{M_2}(\psi_2(x),\psi_2(y))|,
\ee
one defines the heights
\be 
\label{thm-subdiffeo-4}
\aligned
h & =\sqrt{\lambda ( \max\{D_{W_1},D_{W_2}\} +\lambda/4 )\,}, 
\\  
\bar{h}&= \max\{h,  \sqrt{\eps^2 + 2\eps} \; D_{W_1}, \sqrt{\eps^2 + 2\eps} \; D_{W_2} \}.
\endaligned
\ee
Then, the intrinsic flat distance between the settled completions is bounded as follows: 
\begin{eqnarray*}
d_{\mathcal{F}}(M'_1, M'_2) &\le&
\left(2\bar{h} + a\right) \Big(
\vol_m(W_{1})+\vol_m(W_2)+\vol_{m-1}(\partial W_{1})+\vol_{m-1}(\partial W_{2})\Big)\\
&&+\vol_m(M_1\setminus W_1)+\vol_m(M_2\setminus W_2).
\end {eqnarray*}
\end{theorem}

\begin{proof}[Proof of Proposition~\ref{compIF}] 
By Proposition~\ref{prop-int-curr-space}, for
$j= 1, \ldots, +\infty$ we have that $U_j$ 
is an
integral current space when viewed as a metric space with the restricted metric $d_{g_j}$ and  whose current structure is defined by (\ref{T}).   
The metric $g_j$ which is defined using the profile function $f_j$
may also be defined using the backwards profile functions $h_j$
so that
$$
g_j= d\sigma^2 + h_j(\sigma) g_{\mathbb{S}^2}
$$
with $\sigma \in [0,D_0]$ where we have extended $h_j$ as a constant
to reach $D_0$ if needed.   Recall that $h_j(0)=r_0>0$ for all 
$j\in \{1,..,\infty\}$. 

For any $\varepsilon>0$, let
$$
\aligned
W_{j,\varepsilon}& = \{(\sigma,\theta)\in U_j: \, \sigma\in A_{j,\varepsilon} \},
\\
W_{\infty, j,\varepsilon}& = \{(\sigma,\theta)\in U_\infty: \, \sigma\in A_{j,\varepsilon} \},
\endaligned
$$
where
$$
A_{j,\varepsilon}= \left\{\sigma\in [0,D_0]: \, \frac{h_j(\sigma)}{h_\infty(\sigma)} \in (1+\eps)^{-2},(1+\eps)^2\right\}.
$$
In particular $h_\infty(\sigma)$ and $h_j(\sigma)$ are positive
for $(\sigma,\theta)\in W_{j,\varepsilon}$.  
Then, we have 
$$
\aligned
D_{W_{j,\varepsilon}}
& = \sup\{\diam_{U_j}(W): \, W\textrm{ is a connected component of } W_{j,\varepsilon}\} \le \diam(U_j)\le r_0+D_0,
\\
D_{W_{\infty,j,\varepsilon}}
& = \sup\{\diam_{U_\infty}(W): \, W\textrm{ is a connected component of } W_{\infty,j,\varepsilon}\} \le \diam(U_\infty)
\le r_0+D_0
\endaligned
$$
Observe that
\be
a=a_\varepsilon= 2\frac{\arccos(1+\varepsilon)^{-1} }{\pi}(r_0+D) \to 0
\quad 
\textrm{ as } \varepsilon \to 0.
\ee
Now, we have 
\be
\lambda=\sup\left\{
\left|d_{U_j}((\theta_1,\sigma_1),(\theta_2,\sigma_2))
-d_{U_\infty}((\theta_1,\sigma_1),(\theta_2,\sigma_2))\right|
 :\,\, \sigma_1, \sigma_2\in A_{j,\varepsilon} \, \theta_1,\theta_2\in \mathbb{S}^2\right\}\le \lambda_{j}, 
\ee
where
\be
\lambda_{j}:=\sup\left\{
\left|d_{U_j}((\theta_1,\sigma_1),(\theta_2,\sigma_2))
-d_{U_\infty}((\theta_1,\sigma_1),(\theta_2,\sigma_2))\right|
 :\,\, \sigma_1, \sigma_2\in [0,D_0] \, \theta_1,\theta_2\in \mathbb{S}^2\right\}. 
\ee

We state a separate result in Lemma~\ref{lem-lambda} below, which show us that
$
\lim_{j\to \infty}\lambda_j =0.
$
We can then define the heights as in (\ref{thm-subdiffeo-4}) and obtain
$\bar{h}_{j}$
such that
\be
\bar{h}_{j}\to 0 \textrm{ whenever } \lambda_{j}\to 0.
\ee
Then the intrinsic flat distance is bounded, since 
\begin{eqnarray*}
d_{\mathcal{F}}(U_j, U_\infty) &\le&
\left(2\bar{h}_{j} + a_\varepsilon\right) \Big(
\vol_m(W_{j\varepsilon})+\vol_m(W_{\infty,j,\varepsilon})+\vol_{m-1}(\partial W_{j\varepsilon})+\vol_{m-1}(\partial W_{\infty,j,\varepsilon})\Big)\\
&&+\vol_m(U_j\setminus W_{j\varepsilon})+\vol_m(U_\infty\setminus W_{\infty,j,\varepsilon})\\
&\le&
\left(2\bar{h}_{j} + a_\varepsilon\right) (4\pi r_0^2 D_0 + 4\pi r_0^2 D_0 + 8\pi r_0^2 + 8\pi r_0^2\Big)\\
&& +\vol_m(U_j\setminus W_{j\varepsilon})+\vol_m(U_\infty\setminus W_{\infty,j,\varepsilon}).
\end {eqnarray*}

Now, we take $\delta>0$, and let 
$$
\sigma_\delta=\inf\{\sigma: h_\infty(\sigma)<\delta\}.
$$
Then by the pointwise convergence,
 for $j$ sufficiently large depending on $\delta$,
$$
h_j(\sigma_\delta) \in (\delta/2, 2\delta)
$$
and, so, by the monotonicity
$$
\aligned
h_\infty(\sigma)>\delta>\delta/2 & h_j(\sigma) > \delta/2 && \textrm{ on } [0,\sigma_\delta],
 \\
h_\infty(\sigma)<\delta<2\delta & h_j(\sigma) < 2\delta && \textrm{ on } [\sigma_\delta, D_0].
\endaligned
$$
Thus, we deduce that 
\be
\aligned
\vol_m(U_j \setminus W_{j\varepsilon})&\le V_{j,\eps,\delta},
\\
\vol_m(U_\infty \setminus W_{\infty,j,\varepsilon})&\le V_{j,\eps,\delta}
\endaligned
\ee
where
$$
V_{j,\eps,\delta}= D_0 4\pi \delta^2 + 
\mathcal{L}\left(A^C_{j\varepsilon}\cap [0, \sigma_\delta] \right) 4\pi r_0^2.
$$
If
$\sigma \in A^C_{j\varepsilon}\cap [0, \sigma_\delta]$, 
then
$$
h_j(\sigma)\ge (1+\varepsilon)^2 h_\infty(\sigma)\ge h_\infty(\sigma)+\varepsilon^2 \delta/2
$$
or
$$
h_\infty(\sigma)\ge (1+\varepsilon)^2 h_j(\sigma)\ge h_j(\sigma)+\varepsilon^2\delta/2 
$$
and, in either case,
\be
\inf\left\{|h_j(\sigma)-h_\infty(\sigma)|: \, \sigma \in  
A^C_{j\varepsilon}\cap [0, \sigma_\delta] \right\} \ge \varepsilon^2\delta/2.
\ee

Now, we obtain 
$$
\aligned
\int_0^{D_0} |h_j(\sigma)-h_\infty(\sigma)|^2 \, d\sigma
&\ge 
\int_{A^C_{j\varepsilon}\cap [0, \sigma_\delta] } |h_j(\sigma)-h_\infty(\sigma)|^2 \, d\sigma \\
&\ge  
\mathcal{L}\left(A^C_{j\varepsilon}\cap [0, \sigma_\delta] \right) \varepsilon^2\delta/2.
\endaligned
$$
Recall that in Theorem~\ref{compSobolev} we proved the convergence 
 $h_j \to h_\infty$ in $L^2[0,D_0]$.   So, 
for fixed $\varepsilon>0$ and $\delta>0$, we have
\be
\lim_{j\to \infty} \mathcal{L}\left(A^C_{j\varepsilon}\cap [0, \sigma_\delta] \right) \to 0
\ee
and thus
\be
\lim_{j\to\infty} V_{j,\eps,\delta} = D_0 4\pi \delta^2. 
\ee
Thus, for the flat distance, 
\begin{eqnarray*}
\lim_{j\to \infty}d_{\mathcal{F}}(U_j, U_\infty) 
&\le&
\left(2\lim_{j\to\infty}\bar{h}_{j} + a_\varepsilon\right) (8\pi r_0^2 D_0 + 16\pi r_0^2\Big) + D_0 4\pi \delta^2 \\
&\le&
\left(0 + a_\varepsilon\right) (8\pi r_0^2 D_0 + 16\pi r_0^2\Big) + D_0 4\pi \delta^2. 
\end {eqnarray*}
Taking $\delta \to 0$ and then $\eps \to 0$ we have completed the proof of \eqref{eq-compIF}.

Finally, we can apply Theorem~\ref{compSobolev} to this sequence and we see that
$h_j(0)\to h_\infty(0)$ implies (\ref{c-area}) while
$$
\vol(U_j)=\int_0^{D_0} h_j(\sigma)\, d\sigma \to \int_0^{D_0} h_\infty(\sigma)\, d\sigma  
=\vol(U_\infty)
$$
implies (\ref{c-volume}).  Note that, in general, intrinsic flat convergence only implies 
lower semicontinuity of the mass; yet, here, we have continuity and the mass agrees with
the volume and the area.   The convergence of the Hawking mass claimed in
(\ref{c-Hawking}) also follows from Theorem~\ref{compSobolev}.

Finally, we establish the bound (\ref{c-depth}), as follows.  Let
$D_1=\liminf_{j \to +\infty} \Depth(\Sigma_j) \le D_0$. 
If $D_1=D_0$ then we are done since  $\Depth(\Sigma_\infty)\le D_0$
by the definition of $U_\infty$ in Theorem~\ref{compSobolev}.   If $D_1<D_0$
then, by the definition of liminf,
$$
 \text{for all } D_2\in (D_1, D_0],\,\text{there exists } N_{D_2}
\textrm{ such that }
\sup_{j\ge N_{D_2}} \Depth(\Sigma_j)\le D_2,
$$
and so
$$
h_j(\sigma)=0 \textrm{ for } \sigma \in (D_2, D_0].
$$
Taking $j\to \infty$ we also have
$$
h_\infty(\sigma)=0 \textrm{ for } \sigma \in (D_2, D_0]
$$
and so
$$
\Depth(\Sigma_\infty) \le D_2.
$$
Taking $D_2 \to D_1$ we obtain
$\Depth(\Sigma_\infty) \le D_1$
and we are done.
\end{proof}

Finally we stated and prove the promised lemma. 
 
\begin{lemma}\label{lem-lambda}
If $h_j \to h_\infty$ 
and
\be
\lambda_{j}=\sup\left\{
\left|d_{U_j}((\theta_1,\sigma_1),(\theta_2,\sigma_2))
-d_{U_\infty}((\theta_1,\sigma_1),(\theta_2,\sigma_2))\right|
 :\,\, \sigma_1, \sigma_2\in [0,D_0] \, \theta_1,\theta_2\in \mathbb{S}^2\right\},
\ee
then
\be
\lim_{j\to\infty} \lambda_j=0.
\ee
\end{lemma}

\begin{proof} We proceed by contradiction. Then, there exists $k_0>0$
and $(\theta_{1j},\sigma_{1j}),(\theta_{2j},\sigma_{2j})\in \mathbb{S}^2\times[0,D_0]$ 
 such that
\be
|d_{U_j}((\theta_{1j},\sigma_{1j}),(\theta_{2j},\sigma_{2j})) 
-d_{U_\infty}((\theta_{1j},\sigma_{1j}),(\theta_{2j},\sigma_{2j})) |\ge k_0.
\ee
Let 
$$
\delta_j=\sup \big\{ |(h_j(\sigma))^2-(h_\infty(\sigma))^2|^{1/2}: \, \sigma\in [0,D_0] \big\}
$$
and recall that in Theorem~\ref{compSobolev} we have proven $\delta_j\to 0$.
This implies that the lengths of curves converge as described in Remark~\ref{lengths}.

Recall that between any pair of points, there is a curve $(\theta(t), \sigma(t))$
whose length is the distance between those points.   When the metric is rotationally
symmetric, then $\sigma$ is in fact a reparametrized geodesic in $\mathbb{S}^2$.   
Any longer path taken by $\sigma$ would only make the length longer.

Now suppose we have a curve $(\theta(t), \sigma(t))$ running from
$(\theta_{1j},\sigma_{1j})$ to $(\theta_{2j},\sigma_{2j}))$
such that
\be
d_{U_\infty}((\theta_{1j},\sigma_{1j}),(\theta_{2j},\sigma_{2j})) =\int_0^1 \sqrt{|\sigma'(t)|^2 + h_\infty^2(\sigma(t))|\theta'(t)|^2} \, dt. 
\ee
Then, we find 
$$
\aligned
d_{U_j}((\theta_{1j},\sigma_{1j}),(\theta_{2j},\sigma_{2j})) 
&\le\int_0^1 \sqrt{|\sigma'(t)|^2 + h_j^2(\sigma(t))|\theta'(t)|^2} \, dt \\
&\le \int_0^1 \sqrt{|\sigma'(t)|^2 + (h_\infty(\sigma(t))^2+\delta_j^2)(\sigma(t))|\theta'(t)|^2} \, dt \\
&\le \int_0^1 \sqrt{|\sigma'(t)|^2 + (h_\infty(\sigma(t))^2(\sigma(t))|\theta'(t)|^2} \, dt 
 + \delta_j \int_{0}^1|\theta'(t)| \, dt \\
&\le d_{U_\infty}((\theta_{1j},\sigma_{1j}),(\theta_{2j},\sigma_{2j})) +\delta_j \pi.
\endaligned
$$
If on the other hand we take a curve $(\theta(t), \sigma(t))$ running from
$(\theta_{1j},\sigma_{1j})$ to $(\theta_{2j},\sigma_{2j}))$
such that
\be
d_{U_j}((\theta_{1j},\sigma_{1j}),(\theta_{2j},\sigma_{2j})) 
=\int_0^1 \sqrt{|\sigma'(t)|^2 + h_j^2(\sigma(t))|\theta'(t)|^2} \, dt. 
\ee
Then, we have 
$$
\aligned
d_{U_\infty}((\theta_{1j},\sigma_{1j}),(\theta_{2j},\sigma_{2j})) 
&\le\int_0^1 \sqrt{|\sigma'(t)|^2 + h_\infty^2(\sigma(t))|\theta'(t)|^2} \, dt \\
&\le \int_0^1 \sqrt{|\sigma'(t)|^2 + (h_j(\sigma(t))^2+\delta_j^2)(\sigma(t))|\theta'(t)|^2} \, dt \\
&\le \int_0^1 \sqrt{|\sigma'(t)|^2 + (h_j(\sigma(t))^2(\sigma(t))|\theta'(t)|^2} \, dt 
+ \delta_j \int_{0}^1|\theta'(t)| \, dt \\
&\le d_{U_j}((\theta_{1j},\sigma_{1j}),(\theta_{2j},\sigma_{2j})) +\delta_j \pi.
\endaligned
$$
We conclude that 
\be
|d_{U_j}((\theta_{1j},\sigma_{1j}),(\theta_{2j},\sigma_{2j})) 
-d_{U_\infty}((\theta_{1j},\sigma_{1j}),(\theta_{2j},\sigma_{2j})) |\le \delta_j \pi
\ee
and for $j$ sufficiently large we have reached a contradiction.
\end{proof}


\section{Examples}\label{sec:9}

In this section we provide the full details of examples mentioned earlier in this paper. We work in dimension three and, for each example, we
 provide a detailled construction.  An approach for constructing these examples is to refer to Lemma 2.6 in  \cite{LS1} by Lee and the second author. Therein, it was pointed out that given any smooth increasing function $\overline{m}_H:[0,\infty)\to[0,\infty)$ such that $\overline{m}_H(0)=0$ and
\be\label{admissible}
\overline{m}_H(r) < \frac{1}{2}r \qquad \text{ for all } r>0, 
\ee
there exists a smooth rotationally symmetric 3 dimensional Riemannian manifold with metric
\be\label{metric-constr}
g=(1+[z'(r)]^2)dr^2 + r^2 g_0,
\ee
with nonnegative scalar curvature such that 
the Hawking mass of the level set $z^{-1}(r)$ coincides with the prescribed function $\overline{m}_H(r)$.
Specifically, we find 
\be \label{lem-ex-m-H-to-z}
 z(\bar{r}) = \int_{\rmin}^{\bar{r}} \sqrt{ \frac{2\mathrm{m}_{\mathrm{H}}(r)}
{r^{m-2}-2\mathrm{m}_{\mathrm{H}}(r)}}\,dr,
\ee
and 
so $z(0)=0$.   

\begin{exampleproposition}\label{no-scalar}
There exist sequences of manifolds $M_j \in \RS_m^\reg$ satisfying the  uniform bounds in 
Theorem~\ref{compactness} which converge to Euclidean space $\mathbb{E}^3$ and have
$m_{ADM}(M_j)\to 0$, but $Scalar(p_j)\to K_\infty\in (0,\infty]$.  Hence, the scalar curvature
need not converge and, on the other hand, the points $p_j$ may lie on the pole or on $\Sigma_j$.
\end{exampleproposition}

\begin{proof}
Let $K_j\in (0,\infty)$ be increasing such that $K_j \to K_\infty$.   Let $\delta_j\in (0,\infty)$
decrease to $0$.   In this first example we take $p_j$ to be the poles and set $m_H(r)$ to be the Hawking mass
of (\ref{metric-constr}) where 
\be
z(r)=\sqrt{1/K_j^2 -r^2} -1/K_j
\ee
for $r\in [0,\delta_j]$ which is increasing since $Scalar = K_j\ge 0$.   
For $r\ge \delta_j$ set
\be
m_H(r)=m_H(\delta_j)
\ee
so that it continues to be nondecreasing.    If we choose $\delta_j$ decreasing to $0$
fast enough that $m_H(\delta_j)\to 0$ then $m_{ADM}(M_j)\to 0$ and the intrinsic flat limit
is Euclidean space.   

Next we take $m_j \to 0$ and
$p_j\in \Sigma_j$ such that $\Area(\Sigma_j)=A_0=4\pi r_0^2$.
Set $m_H(r)=m_j$ for $r\in [2m_j, r_0-\delta_j]$.   
This gives us a $z(r)$ defined up to $r=r_0-\delta_j$.
Let
\be
a_j=z(r_0-\delta_j) \textrm{ and } m_j=z'(r_0-\delta_j).
\ee
Choose $b_j>0$ such that the circle about $(0,b_j)$ of radius $1/K_j$ touches
the point $(r_0-\delta_j, a_j)$ with a tangent line of slope $m_j$.   
Let
\be
z(r)=\sqrt{1/K_j^2 - r^2} + b_j \textrm{ for } r\in [r_0-\delta_j, r_0+\delta_j].
\ee
For $r\ge r_0+\delta_j$, we set 
$m_H(r)=m_H(r_0+\delta_j)$, 
so that it continues to be nondecreasing.    If we choose $m_j$ and $\delta_j$ decreasing to $0$
fast enough that $m_H(r_0+\delta_j)\to 0$ then $m_{ADM}(M_j)\to 0$ and the intrinsic flat limit
is Euclidean space.   However $Scalar_{p_j}=K_j\to K_\infty$.
\end{proof}

The next example first appeared in \cite{LS1} demonstrating why Gromov-Hausdorff
convergence fails to provide stability of the positive mass theorem and why we needed
to study intrinsic flat convergence.   Here we make this example more explicit and show that it justifies why
we are 
studying backwards profile functions in Theorem~\ref{compSobolev}, why reversed backwards
limit profile functions are not the limits of profile functions and why we only obtain semicontinuity 
of the depth function in Theorem~\ref{compactness}:

\begin{exampleproposition}\label{thin-well}
There exist sequences of manifolds,
$M_j \in \RSone$ satisfying the conditions of
Theorem~\ref{compactness} which converge in the intrinsic flat sense
to Euclidean space $\mathbb{E}^3$ and have mass 
$m_{ADM}(M_j)\to 0$ which have increasingly thin wells
such that the reversed backwards limit profile function does not agree with
the limit of the profile functions and such that 
\be
\Depth(\Sigma_\infty)=3< \lim_{j\to \infty} \Depth(\Sigma_j)=6. 
\ee
\end{exampleproposition}

\begin{proof}
We want to construct a precise sequence of metrics $g_j$ with a very thin deep well.
Let $L>0$ and let 
\be
z_j(r)= L(jr)^2 \textrm{ for } r\in [0,1/j]
\ee
so that by (\ref{Hawking-r-bound}) we have (in dimension three) 
\be
m_{Hj}(r)=\frac{r}{2} \frac{(2Lj^2r)^2}{1+(2Lj^2r)^2} 
\ee
and thus 
\be
H_j:=m_{Hj}(1/j) =\frac{(1/j)}{2} \frac{(2Lj)^2}{1+(2Lj)^2}\to 0
\textrm{ as } j\to \infty.
\ee
We then prescribe 
\be
m_{Hj}(r)=H_j \textrm{ for } r\ge 1/j
\ee
and define $z_j(r)$ as in (\ref{lem-ex-m-H-to-z}) with
$r_{min}=0$.   Observe that
\be
z_j(r) = L + \sqrt{2H_j(r-2H_j)}- \sqrt{2H_j((1/j)-2H_j)}
 \textrm{ for } r\ge 1/j.
\ee
To see that we satisfy the conditions for Theorem~\ref{compactness}
we take $\Sigma_j\subset M_j \textrm{ be } r^{-1}(3)$, 
so that
\be
\Area(\Sigma_j)=4\pi 3^2, 
\qquad
 m_H(\Sigma_0)=H_j \le 1
\ee 
for $j$ sufficiently large.
In addition, we find 
$$
\aligned
\Depth(\Sigma_j)&= d_{g_j}(r^{-1}(3), r^{-1}(0))
=d_{g_j}(r^{-1}(1/j), r^{-1}(0))+d_{g_j}(r^{-1}(3), r^{-1}(1/j))
\\
&=\int_{0}^(1/j) \sqrt{ 1+z'(r)^2 } \, dr+
\int_{(1/j)}^3 \sqrt{ 1+z'(r)^2 } \, dr\\
&\le  \int_{0}^(1/j)  1+|z'(r)|\, dr+
\int_{(1/j)}^3  1+|z'(r)| \, dr\\
&\le  (1/j) + z(1/j)-z(0) + 3 + z(3)-z(1/j), 
\endaligned
$$
thus 
\be
\Depth(\Sigma_j) \le  (1/j) + L  + 3 + \sqrt{2H_j(3-2H_j)}, 
\ee
which is uniformly bounded so that we satisfy the conditions
of Theorem~\ref{compactness}.

Note also that, for the depth, 
$$
\aligned
\Depth(\Sigma_j)&= d_{g_j}(r^{-1}(3), r^{-1}(0))
=d_{g_j}(r^{-1}(1/j), r^{-1}(0))+d_{g_j}(r^{-1}(3), r^{-1}(1/j))
\\
&=\int_{0}^{(1/j)} \sqrt{ 1+z'(r)^2 } \, dr+
\int_{(1/j)}^3 \sqrt{ 1+z'(r)^2 } \, dr\\
&\ge  \int_{0}^{(1/j)}  |z'(r)|\, dr+
\int_{(1/j)}^3  1 \, dr \ge  z(1/j)-z(0) + 3 \ge  L +3 
\endaligned
$$
and thus
$
\lim_{j\to \infty} \Depth(\Sigma_j)=L+3, 
$
in which $L$ was arbitrary.

Since $m_{ADM}(M_j) =H_j \to 0$, we know by the stability of the positive mass theorem \cite{LS1} that $M_j$ converge to Euclidean space in the intrinsic flat sense. Thus, we have 
\be
\Depth(\Sigma_\infty)=\Depth\left(r^{-1}(3)\subset \E^3\right)=3
\ee
and thus
\be
\lim_{j\to \infty} \Depth(\Sigma_j)=L+3 > \Depth(\Sigma_\infty).
\ee

In addition by, Theorem~\ref{thm-sobolev-no-diffeo},
the backward profile functions $h_j(\sigma)$ must converge to 
$h_{\E}(\sigma)$ and so the reversed backwards limit profile function defined in Theorem~\ref{compSobolev} is $f(s)=s$ as in Euclidean space.

But let us examine exactly what happens to the ordinary profile functions, $f_j(s)$,
where $s$ is the distance from $r^{-1}(0)$.   We have $f_j(0)=0$.  Then there exist points 
$$
\aligned
s_{0,j}&=d_{M_j}(r^{-1}(0), r^{-1}(3)) =\Depth(\Sigma_j),
\\
s_{1,j}&=d_{M_j}(r^{-1}(0), r^{-1}(1/j), 
\endaligned
$$
so that
\be
f_j(s_{0,j})=3, \qquad  f_j(s_{1,j})=1/j.
\ee
We then observe that
$$
\aligned
s_{1,j}&=d_{g_j}(r^{-1}(1/j), r^{-1}(0))\\
&\ge  \int_{0}^(1/j)  |z'(r)|\, dr \ge z(1/j)-z(0) \ge L.
\endaligned
$$
Thus $[0, L]\subset [0, s_{1,j}]$, and so
\be
f(s) \le 1/(j) \quad \textrm{ for all } s\in [0, L], 
\ee
and, therefore, these profile functions converge to $f_\infty(s)=0$ on $[0, L]$.  Thus $f_\infty(s)\neq f(s)$.    

The backward profile functions, $h_j(\sigma)=f_j(s_{0,j}-\sigma)$ are
well controlled since they are based on the level set $h_j^{-1}(0)=\Sigma_0\subset M_j$, 
which persists in the intrinsic flat limit, while the profile functions $f_j$ vanish in the limit since 
they are based at a point which is ``disappearing'' in the limit.
\end{proof}

\begin{exampleproposition}
There exist sequences of manifolds $M_j\in\RS_m^\reg$ satisfying the uniform bounds in 
 Theorem~\ref{compactness} that converge in the intrinsic
flat and Sobolev sense toward a limit $M_\infty \in \RSonebar\setminus \RSone$.
\end{exampleproposition}

\begin{proof}
It is easy to construct a sequence  of smooth functions $f_j$
which approaches (for instance) 
$$
f_\infty(s)
=\begin{cases}
\sin(s),            & s\in [0,\pi/2], 
\\
1      &  s\in [\pi/2,\pi]
\end{cases}
$$ 
and $f_\infty(s)$ defined for $s>\pi$ to have
constant Hawking mass equal to $2$.  In fact
we can consider any sequence satisfying $f_j(s)=f_\infty(s)$ for $s>\pi$ and
while $f_j'(s)>0$ and $f_j"(s)>0$ for $s<\pi$ such that $f_j$ 
 converges in the $C^1$ norm toward $f_\infty$.  Such functions 
$f_j$ are suitable profile functions for defining the sequence of spaces $M_j$.
In this example, $f_\infty(s)$ agrees
with the reversed backwards limit profile function from
$\Sigma_\infty=r^{-1}(2)$, since the manifolds are
smoothly converging.
\end{proof}



\end{document}